\documentclass[a4paper,USenglish]{lipics-v2019}

\usepackage{amsthm}
\usepackage{amsmath}
\usepackage{nccmath}
\usepackage{amssymb}
\usepackage{bm}
\usepackage{graphicx}
\usepackage{enumitem}
\usepackage{hyperref}
\usepackage{array}
\usepackage{makecell}
\usepackage{multirow}

\newcommand{\deleted}[1]{}

\usepackage[disable]{todonotes}
\newcommand{\jr}[1]{\todo[color=white!40!white]{\scriptsize #1}{}}

\newcommand{\R}{\mathbb{R}}
\newcommand{\E}{\mathcal{E}}
\newcommand{\A}{\mathcal{A}}
\newcommand{\T}{T}
\newcommand{\VD}{\mathrm{VD}}
\newcommand{\AD}{\mathrm{AD}}
\newcommand{\RD}{\mathrm{RD}}
\newcommand{\SC}{\mathrm{SC}}
\newcommand{\RDS}{\mathrm{RDS}}
\newcommand{\func}[1]{{\ensuremath{\mathsf{#1}}}}
\newcommand{\polylog}{\func{polylog}}

\title{Nearly Optimal Planar $k$ Nearest Neighbors Queries under General Distance Functions}
\titlerunning{Nearly Optimal Planar $k$ Nearest Neighbors Queries under General Distance Functions}

\author{Chih-Hung Liu}{Department of Computer Science, ETH Z\"{u}rich, Z\"{u}rich, Switzerland}{chih-hung.liu@inf.ethz.ch}{https://orcid.org/0000-0001-9683-5982}{}
\authorrunning{Chih-Hung Liu}
\Copyright{Chih-Hung Liu}

\category{}

\nolinenumbers
\hideLIPIcs 

\begin{document}

\maketitle
\begin{abstract}
We study the \emph{$k$ nearest neighbors} problem in the plane for general, convex, pairwise disjoint sites of constant description complexity such as line segments, disks, and quadrilaterals and with respect to a general family of distance functions including the $L_p$-norms and additively weighted Euclidean distances.
For point sites in the Euclidean metric, after four decades of effort, an optimal data structure has recently been developed with $O(n)$ space, $O(\log n+k)$ query time, and $O(n\log n)$ preprocessing time~\cite{AfshaniC09,ChanT16}.
We develop a static data structure for the general setting with nearly optimal $O(n\log\log n)$ space, the optimal $O(\log n+k)$ query time, and expected $O(n\;\polylog\;n)$ preprocessing time. 
The $O(n\log\log n)$ space approaches the linear space, whose achievability is still unknown with the optimal query time, and improves the so far best $O\big(n(\log^2n)(\log\log n)^2\big)$ space of Bohler et~al.'s work~\cite{BohlerKL19}. 
Our dynamic version (that allows insertions and deletions of sites) also reduces the space of Kaplan et~al.'s work~\cite{KaplanMRSS17} from $O(n\log^3 n)$ to $O(n\log n)$ while keeping $O(\log^2 n+k)$ query time and $O(\polylog\;n)$ update time, thus improving many applications such as
dynamic bichromatic closest pair and dynamic minimum spanning tree in general planar metric, and shortest path tree and dynamic connectivity in disk intersection graphs.

To obtain these progresses, we devise \emph{shallow cuttings of linear size} for general distance functions. Shallow cuttings are a key technique to deal with the $k$ nearest neighbors problem for point sites in the Euclidean metric. 
Agarwal et~al.~\cite{AgarwalES99} already designed linear-size shallow cuttings for general distance functions, but their shallow cuttings could not be applied to the $k$ nearest neighbors problem. 
Recently, Kaplan et~al.~\cite{KaplanMRSS17} constructed shallow cuttings that are feasible for the $k$ nearest neighbors problem, while the size of their shallow cuttings has an extra double logarithmic factor.
Our innovation is a new random sampling technique for the analysis of geometric structures. 
While our shallow cuttings seem, to some extent, merely a simple transformation of Agarwal et~al.'s~\cite{AgarwalES99}, the analysis requires our new technique to attain the linear size. 
Since our new technique provides a new way to develop and analyze geometric algorithms, we believe it is of independent interest.
\end{abstract}

\section{Introduction}\label{sec-ind}

Dating back to Shamos and Hoey (1975)~\cite{ShamosH75}, the \emph{$k$ nearest neighbors} problem is one fundamental problem in computer science:
given a set $S$ of $n$ geometric \emph{sites} in the plane and a distance measure,
build a data structure that answers for a query point $p$ and a query integer $k$, the $k$ nearest sites of $p$ in $S$.
A related problem called \emph{circular range query} problem is instead to answer for a query point $p$ and a query radius $\delta$,
all the sites in $S$ whose distance to $p$ is at most $\delta$.
A circular range query can be answered through $k$ nearest neighbors queries for $k=\log n, 2\log n, 4\log n,\ldots$ until all the sites inside the circular range have been found, i.e., until one found site is not inside the circular range (\cite[Corollary~2.5]{Chan00}).
For point sites in the Euclidean metric,
these two problems have received considerable attention in theoretical computer science~\cite{AfshaniC09,AggarwalHL90,BentleyM79,Chan00,Chan10,ChanT16,ChazelleCPY86,ColeY84,KaplanMRSS17,Matousek92a,Ramos99a,ShamosH75}.
Many practical scenarios, however, entail non-point sites and non-Euclidean distance measures, which has been extensively discussed by Kaplan et~al.~\cite{KaplanMRSS17}.
Therefore, for  practical applications,
it is beneficial to study the $k$ nearest neighbors problem for general distance functions.

The key technique for point sites in the Euclidean metric is \emph{shallow cuttings}, a notion to be defined later.
Agarwal et~al.~\cite{AgarwalES99} already generalized shallow cuttings to general distance functions, but their shallow cuttings could not be applied to the $k$ nearest neighbors problem.
Recently, Kaplan et~al.~\cite{KaplanMRSS17} constructed shallow cuttings that are feasible for the $k$ nearest neighbors problem, while the size of their shallow cuttings has an extra double logarithmic factor. 
Our main contribution is to devise \emph{linear-size} shallow cuttings for the $k$ nearest neighbors problem under general distance functions, shedding light on achieving the same complexities as point sites in the Euclidean metric.

Based on our linear-size shallow cuttings, 
we build a \emph{static} data structure for the $k$ nearest neighbors problem with nearly optimal $O(n\log\log n)$ space and the optimal $O(\log n+k)$ query time.
The $O(n\log\log n)$ space approaches the linear space, whose achievability is still unknown,
and improves the so far best $O\big(n(\log^2n)(\log\log n)^2\big)$ space of  Bohler et~al.'s work~\cite{BohlerKL19}.
Our shallow cuttings also enable a \emph{dynamic} data structure that allows insertions and deletions of sites with $O(n\log n)$ space,
improving the $O(n\log^3 n)$ space of Kaplan et~al.'s work~\cite{KaplanMRSS17}.

Our innovation is a new random sampling technique for the analysis of geometric structures.
While our shallow cuttings seem, to some extent, merely a simple transformation of Agarwal et~al.'s~\cite{AgarwalES99}, to attain the linear size, the analysis requires our new technique to deal with \emph{global and local conflicts of a configuration}. 
For example, to compute a triangulation for $n$ points, 
a \emph{configuration} is a triangle defined by three points, and a point is said to \emph{conflict with} a triangle if the point lies inside the triangle. 
\emph{Global} and \emph{local} conflicts of a triangle are associated respectively with all the $n$ points and a random subset.
Our technique employs relatively many \emph{local} conflicts to prevent relatively few \emph{global} conflicts, in contradistinction to many state-of-the-art techniques that adopt relatively many \emph{global} conflicts to prevent zero \emph{local} conflict.
This conceptual difference enables our technique to directly analyze local geometric structures; for a simple illustration, see Section~\ref{sub-ind-contribution}.

Each site in $S$ can be represented as the graph of its distance function, namely an $xy$-monotone surface in $\R^3$ where the $z$-coordinate is the distance from the $(x,y)$-coordinates to the respective site.
For example, the surface for a point site $(a, b)$ in the $L_1$ norm is the inverted pyramid $z=|x-a|+|y-b|$.
By this interpretation, the $k$ nearest sites of a query point $p$ become the $k$ lowest surfaces along the vertical line passing through $p$. 
In this paper, we restrict to the case that the lower envelope of any $r$ surfaces has $O(r)$ faces, edges and vertices, as Kaplan et~al.~\cite{KaplanMRSS17} pointed out that this restriction works for many applications.

For point sites in the Euclidean metric, instead of the above interpretation, a standard lifting technique can map each point site to a plane tangent to the unit paraboloid $z=\bm{-}(x^2+y^2)$~\cite{Matousek02}, so that the $k$ nearest point sites of a query point become the $k$ lowest planes along the vertical line passing through the query point.
An optimal data structure for the $k$ lowest plane problem has recently been developed with $O(n)$ space, $O(\log n+k)$ query time, and $O(n\log n)$ preprocessing time \cite{AfshaniC09,ChanT16}.
The dynamic version allows $O(\log^2 n+k)$ query time, amortized $O(\log^3n)$ insertion time, and amortized $O(\log^4n)$ deletion time \cite{Chan10,ChanT16,Chan19}.

\subparagraph{Shallow Cuttings.}
Let $H$ be a set of $n$ planes in $\R^3$,
and define the \emph{level} of a point in $\R^3$ as the number of planes in $H$ lying vertically below it
and the \emph{$(\leq l)$-level} of $H$ as the set of points in $\R^3$ with level of at most $l$. 
An \emph{$l$-shallow $\frac{1}{r}$-cutting} for $H$ is a set of disjoint downward semi-unbounded vertical triangular prisms
covering \emph{the $(\leq l)$-level} of $H$ such that each prism intersects at most $\frac{n}{r}$ planes; see Fig.~\ref{fig-level-prism}.
We abbreviate the $\frac{n}{r}$-shallow $O(\frac{1}{r})$-cutting as \textbf{\emph{$\bm{\frac{1}{r}}$-shallow-cutting}}.
Since a $\frac{1}{r}$-shallow-cutting covers the $(\leq\frac{n}{r})$-level of $H$, the triangular face of each prism lies the above $(\leq\frac{n}{r})$-level, so that for any vertical line through a prism, its $\frac{n}{r}$ lowest planes intersect the prism.
Since each prism stores the $O(\frac{n}{r})$ planes intersecting it, if $\frac{n}{2r}< k\leq \frac{n}{r}$, the $k$ lowest planes of a query vertical line can be answered by locating the prism intersected by the line and selecting the $k$ lowest planes from the $O(\frac{n}{r})$ stored planes, leading to $O(\log n+\frac{n}{r})=O(\log n+k)$ query time.
To cover the whole range of $k$,
$\frac{1}{r}$-shallow-cuttings for $r=2, 4, 8, \ldots, \frac{n}{\log n}$ are sufficient.

\begin{figure}
	\centering
	\includegraphics[width=12cm]{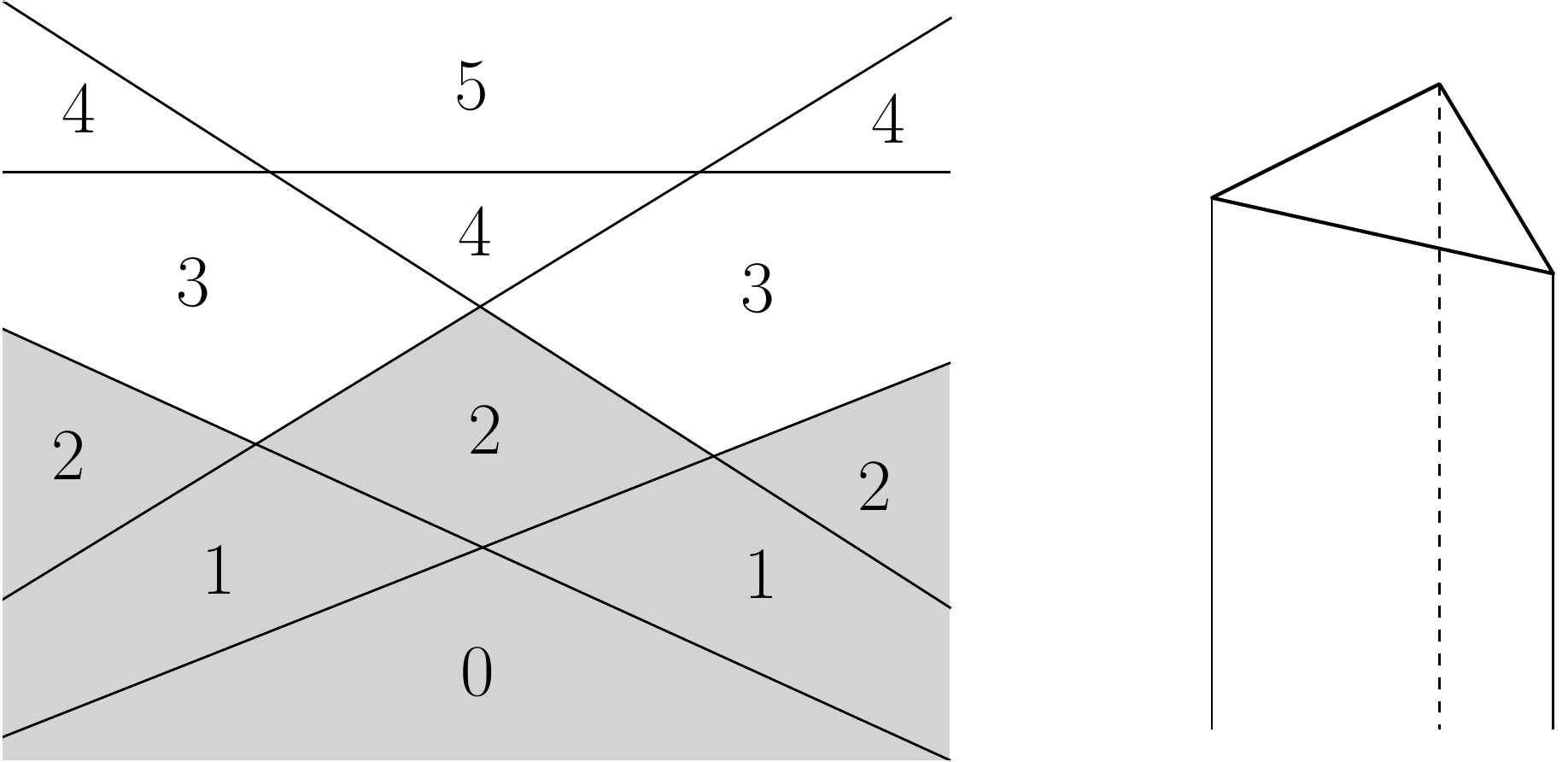}
	\caption{Left: The gray area is the $\leq 2$ level, where the lines depict the planes and the numbers show the levels. Right: A downward semi-unbounded vertical triangular prism.}\label{fig-level-prism}
\end{figure}

Matou\v{s}ek~\cite{Matousek92a} first used ``tetrahedra'' to define shallow cuttings and proved the existence of  a $\frac{1}{r}$-shallow-cutting of $O(r)$ tetrahedra.
Then, Chan~\cite{Chan00} observed that the tetrahedra can be turned into disjoint downward semi-unbounded vertical triangular prisms, resulting in the above-defined shallow cuttings.
Since each prism in a $\frac{1}{r}$-shallow-cutting stores $O(\frac{n}{r})$ planes, a $\frac{1}{r}$-shallow-cutting requires $O(r)\cdot O(\frac{n}{r})=O(n)$ space, so that the $O(\log n)$ $\frac{1}{r}$-shallow-cuttings, i.e., $r=2, 4, 8, \ldots, \frac{n}{\log n}$,
directly compose a data structure for the $k$ lowest plane problem with $O(n\log n)$ space and $O(\log n+k)$ query time. The further literature about shallow cuttings is sketched in Appendix~\ref{ap-literature}.

Matou\v{s}ek's $\frac{1}{r}$-shallow-cutting construction
picks $r$ planes randomly, builds the \emph{canonical triangulation} for the \emph{arrangement} of the $\bm{r}$ sample planes (Section~\ref{sub-rs-configuration}),
selects all tetrahedra in the triangulation, called \textbf{\emph{relevant}}, that intersect the $(\leq \frac{n}{r})$-level of the $\bm{n}$ input planes,
and if a relevant tetrahedron intersects more than $\frac{n}{r}$ planes, 
refines this ``heavy'' one into smaller ``light'' ones.

\subparagraph{Generalization.} 
Agarwal et~al.~\cite{AgarwalES99} generalized Matou\v{s}ek's construction to general distance functions by replacing the canonical triangulation with the \emph{vertical decomposition of surfaces} (Section~\ref{sub-sc-distance-VD}), and built a $\frac{1}{r}$-shallow-cutting of $O(r)$ ``\emph{pseudo-prisms}.''
Pseudo-prisms, defined in Section~\ref{sub-sc-distance-VD}, can be temporarily viewed as axis-parallel cuboids. 
Their  pseudo-prisms, however, \mbox{vertically overlap}, i.e., a vertical line would intersect more than one pseudo-prism, and there is no known efficient method to locate the topmost pseudo-prism intersected by a query vertical line.
Therefore, their shallow cuttings are not suitable for the $k$ nearest neighbors problem.

Recently, Kaplan et~al.~\cite{KaplanMRSS17} instead adopted \emph{($p$, $\epsilon$)-approximations}~\cite{Har-PeledS11} to design a $\frac{1}{r}$-shallow-cutting of $O(r\log^2 n)$ semi-unbounded pseudo-prisms that do not vertically overlap, yielding a data structure for the $k$ nearest neighbors problem with $O(n\log^3 n)$ space,  $O(\log n+k)$ query time and expected $O(n\log^3n\lambda_{s+2}(\log n))$ preprocessing time, where $\lambda_t(\cdot)$ is the maximum length of a Davenport-Schinzel sequence of order $t$ and $s$ is a constant dependent on the surfaces. 
Their dynamic version (that allows insertions and deletions of sites) achieves $O(\log^2 n+k)$ query time, 
expected amortized $O(\log^5n\lambda_{s+2}(\log n))$ insertion time and expected amortized $O(\log^9n\lambda_{s+2}(\log n))$ deletion time. 
They studied only the case $k=1$, and the general case follows from Chan's idea~\cite{Chan10}.

To achieve $O(n\log n)$ or smaller space,
a $\frac{1}{r}$-shallow-cutting of size $O(r)$ would be required.
Since the size of Kaplan et~al.'s shallow cutting is tight, an attempt would transform the pseudo-prisms in Agarwal et~al.'s shallow cutting~\cite{AgarwalES99} into disjoint downward semi-unbounded ones.
A simple transformation picks the top faces of all pseudo-prisms, computes the upper envelope of those top faces, builds the trapezoidal decomposition~\cite{deBergCKO08,Mulmuley94} for the upper envelope, and extends each trapezoid to a downward semi-unbounded pseudo-prism. 
However, considering pair-wise vertical overlap among the $O(r)$ top faces, a trivial bound for the size of the upper envelope is $O(r^2)$.

An observation to reduce the size is that the pseudo-prisms before the refinement come from the vertical decomposition of  the $r$ \emph{sample} surfaces, i.e., they are defined by \emph{sample} surfaces.
By this observation, the overlap between two top faces can be charged to the higher pseudo-prism, so that the number of charges for a pseudo-prism might only depend on the sample surfaces lying fully below the pseudo-prism. 
Then, if a pseudo-prism lies fully above $t$ sample surfaces, it could be possible to derive a function of $t$ to bound the contribution from the pseudo-prism to the size of the upper envelope. 
Moreover, since a relevant pseudo-prism intersects the $(\leq \frac{n}{r})$-level of the $n$ input surfaces, a relevant pseudo-prism lies fully above at most $\frac{n}{r}$ surfaces. 
Therefore, it is worth to study the probability that a pseudo-prism lies fully above $t$ sample surfaces, but at most above $\frac{n}{r}$ surfaces, i.e., a configuration has many local conflicts, but relatively few global conflicts.

   

\vspace*{-0.09cm}\subparagraph{Other General Results.}  
Agarwal et~al.~\cite{AgarwalM94,AgarwalMS13} studied the \emph{range searching} problem with \emph{semialgebraic sets}.
They considered a set $P$ of $n$ points in $\R^d$, and a collection $\Gamma$ of ranges each of which is a subset of $\R^d$ and is defined by a constant number of constant-degree polynomial inequalities. 
They constructed an $O(n)$-space data structure in $O(n\log n)$ time that for a query range $\gamma\in \Gamma$, reports all the $\kappa$ points inside $\gamma$ within $O(n^{1-\frac{1}{d}}+\kappa)$ time, where $\kappa$ is unknown before the query. 
Their data structure can be applied to the \emph{circular range query} problem by mapping each geometric site to a point site in higher dimensions, e.g. a line segment in $\R^2$ can be mapped to a point in $\R^4$.

Bohler et~al.~\cite{BohlerCKLPZ15} generalized the order-$k$ Voronoi diagram~\cite{AurenhammerKL13,Lee82} to Klein's abstract setting~\cite{Klein89}, which is based on a bisecting curve system for $n$ sites rather than concrete geometric sites and distance measures. 
They also proposed randomized divide-and-conquer and incremental construction algorithms~\cite{BohlerLPZ16,BohlerKL19}.
A combination of their results and Chazelle et~al.'s $k$ nearest neighbors algorithm~\cite{ChazelleCPY86} yields a data structure with  $O(n\log^2 n(\log\log n)^2)$ space, $O(\log n+k)$ query time, and expected $O(n\log^4 n)$ preprocessing time for the \emph{$k$ nearest neighbors} problem.

Agarwal et~al.~\cite{AgarwalAS18} investigated \emph{dynamic} nearest neighbor queries that allows inserting and deleting point sites in a static simple polygon of $m$ vertices. They generalized Kaplan et~al.'s shallow cutting~\cite{KaplanMRSS17} to the geodesic distance functions in a simple polygon. 
The key techniques are an implicit presentation for their shallow cutting and an efficient algorithm for the implicit presentation. Their dynamic data structure requires $O(n\log^3 n\log m +m)$ space and allows $O(\log^2n\log^2m)$ query time, amortized expected $O(\log^5n\log m+\log^4n\log^3m)$ insertion time and amortized expected $O(\log^7n\log m+\log^6n\log^3m)$ deletion time.

\subsection{Our Contributions}\label{sub-ind-contribution}

\subparagraph{Random Sampling.} 
We propose a new random sampling technique (Theorem~\ref{thm-many-local-less-global}) for the \emph{configuration space} (Section~\ref{sub-rs-configuration}).
At a high level, our technique says if the \emph{local} conflict size is large, the \emph{global} conflict size is probably not small, while most existing ones say if the \emph{global} conflict size is large, the \emph{local} conflict size is probably not zero.
More precisely, for a set $S$ of $n$ objects and an $r$-element random subset $R$ of S, we prove that if a configuration in a geometric structure defined by $R$ conflicts with $\bm{t}$ objects in $R$, the probability that it conflicts with \emph{at most} $\frac{n}{r}$ objects in $S$ decreases \emph{factorially} in $\bm{t}$. 
By contrast, many state-of-the-art techniques \cite{AgarwalMS98,ChazelleF90,ClarksonS89,deBergDS1995} show that if a configuration conflicts with \emph{at least} $\bm{t}\frac{n}{r}$ objects in $S$, the probability that it conflicts with \emph{no} object in $R$ decreases \emph{exponentially} in $\bm{t}$.

This conceptual contrast provides a new way to develop and analyze geometric algorithms, so we believe our random sampling technique is of independent interest.
Roughly speaking, to bound the number of \emph{local} configurations satisfying certain properties, by our technique, one could directly make use of \emph{local} configurations.
For example, our technique enables a direct analysis for the expected number of \emph{relevant} tetrahedra.
Since a relevant tetrahedron intersects the $(\leq\frac{n}{r})$-level of the $n$ planes, it lies fully above at most $\frac{n}{r}$ planes.
Our technique implies that if a tetrahedron lies fully above $t$ \emph{sample} planes,
the probability that it lies fully above at most $\frac{n}{r}$ planes is $O(\frac{1}{t!})$.
Since the canonical triangulation of the $r$ sample planes has $O\big(r\cdot (t+1)^{\frac{3}{2}}\big)$ tetrahedra lying fully above $t$ sample planes~\cite{SharirST01},
the expected number of relevant tetrahedra is $O\big( \sum_{t\geq 0} (r\cdot (t+1)^{\frac{3}{2}})\cdot \frac{1}{t!}\big)=O(r)$.

Clarkson~\cite{Clarkson87} already derived a factorial bound similar to ours, but his analysis involves all possible configurations instead of only those in a geometric structure,
so that a direct application could not lead to a linear-size shallow cutting; see Remark~\ref{rm-clarkson}.

\subparagraph{$k$ Nearest Neighbors.} 
We design a $\frac{1}{r}$-shallow-cutting for the $k$ nearest neighbors problem under general distance functions,
and prove its expected size to be $O(r)$,
indicating that for general distance functions, it could still be possible to achieve the same complexities as point sites in the Euclidean metric. 
While our design for a $\frac{1}{r}$-shallow-cutting is quite straightforward, 
the key to attain the linear size lies in the analysis.\jr{rephrase \& mention our technique}
The high-level idea is first to prove that for a relevant pseudo-prism in the vertical decomposition of $r$ sample surfaces, if it lies fully above $t$ sample surfaces, it contributes $O(1+t^4)$ to the expected size of our $\frac{1}{r}$-shallow cutting,
and then to adopt our new random sampling technique to show that the expected number of such relevant pseudo-prisms is $O(\frac{r}{t!})$, leading to a bound of $\sum_{t\geq 0}O(\frac{1+t^4}{t!}r)=O(r)$.

Then, we adopt Afshani and Chan's ideas~\cite{AfshaniC09} to compose our shallow cuttings and Agarwal et~al.'s data structure~\cite{AgarwalMS13} into a \emph{static} data structure for the $k$ nearest neighbors problem under general distance functions with the nearly optimal $O(n\log\log n)$ space and the optimal $O(\log n+k)$ query time, improving the combination of Bohler et~al.'s and Chazelle et~al.'s methods~\cite{BohlerKL19,ChazelleCPY86} by a $\big((\log^2 n)\log\log n\big)$-factor in space.
The preprocessing time is $O(n\log^3n\lambda_{s+2}(\log n))$ for which we modify Kaplan et~al.'s construction algorithm~\cite{KaplanMRSS17} to compute our shallow cuttings; for the constant $s$, see Section~\ref{sub-sc-distance-VD}.
Our data structure works for  point sites in any constant-size algebraic convex distance metric and additively weighted Euclidean distances, 
and for disjoint line segments, disks, and constant-size convex polygons in the $L_p$ norms or under the Hausdorff metric.

Replacing the shallow cuttings in Kaplan et~al.'s \emph{dynamic} data structure~\cite{KaplanMRSS17} with ours attains $O(n\log n)$ space, $O(\log^2n+k)$ query time, expected amortized $O\big(\log^5n\lambda_{s+2}(\log n)\big)$ insertion time, and expected amortized $O\big(\log^7n\lambda_{s+2}(\log n)\big)$ deletion time,
improving their  space from $O(n\log^3 n)$ to $O(n\log n)$ and reducing a ($\log^2 n$)-factor from their deletion time. 
The new dynamic data structure consequently improves many applications mentioned by Kaplan et~al.\ as shown in Table~\ref{tb-direct} and Table~\ref{tb-disk}.
For a detailed explanation, see Appendix~\ref{ap-application}.

\begin{table}
	
	\renewcommand{\arraystretch}{1.2}
	\begin{tabular}{|p{5cm}|l|l|}
		\hline
		Problem & Old Bound~\cite{KaplanMRSS17} & New Bound (ours)\\
		\hline
		\hline
		\multirow{3}{5cm}{Dynamic bichromatic closest pair in general planar metric} & $n\log^3 n$ space,   & $n\log n$ space, \\
		& $\log^{10} n\lambda_{s+2}(\log n)$ insertion,   & $\log^{8} n\lambda_{s+2}(\log n)$ insertion, \\
		& $\log^{11} n\lambda_{s+2}(\log n)$ deletion & $\log^{9} n\lambda_{s+2}(\log n)$ deletion \\
		\hline
		\multirow{2}{5cm}{Minimum Euclidean planar bichromatic matching}  & \multirow{2}{*}{$n^2\log^{11} n\lambda_{s+2}(\log n)$}   & \multirow{2}{*}{$n^2\log^{9} n\lambda_{s+2}(\log n)$}  \\
		&  & \\
		\hline 
		\multirow{2}{5cm}{Dynamic minimum spanning tree in general planar metric} & $n\log^5 n$ space,    & $n\log^3 n$ space, \\
		& $\log^{13} n\lambda_{s+2}(\log n)$ update & $\log^{11} n\lambda_{s+2}(\log n)$ update\\
		\hline
		\multirow{4}{5cm}{Dynamic intersection of unit balls in $\R^3$}  & $n\log^3 n$ space,  & $n\log n$ space,\\
		& $\log^{5} n\lambda_{s+2}(\log n)$ insertion,  & $\log^{5} n\lambda_{s+2}(\log n)$ insertion,\\
		& $\log^{9} n\lambda_{s+2}(\log n)$ deletion, & $\log^{7} n\lambda_{s+2}(\log n)$ deletion, \\
		& queries in $\log^2n$ and $\log^5 n$ & queries in $\log^2n$ and $\log^5 n$\\
		\hline
		\multirow{4}{5cm}{Dynamic smallest stabbing disks}  & $n\log^3 n$ space,  & $n\log n$ space,\\
		& $\log^{5} n\lambda_{s+2}(\log n)$ insertion,  & $\log^{5} n\lambda_{s+2}(\log n)$ insertion,\\
		& $\log^{9} n\lambda_{s+2}(\log n)$ deletion, &  $\log^{7} n\lambda_{s+2}(\log n)$ deletion,\\
		& queries in $\log^5 n$ & queries in $\log^5 n$ \\
		\hline
	\end{tabular}
	\caption{Comparison on direct applications.}\label{tb-direct}
	\renewcommand{\arraystretch}{1}
\end{table}

\begin{table}
	
	\renewcommand{\arraystretch}{1.2}
	\begin{tabular}{|p{5cm}|l|l|}
		\hline
		Problem & Old Bound~\cite{KaplanMRSS17} & New Bound (ours)\\
		\hline
		\hline
		\multirow{2}{5cm}{Shortest path tree in a unit disk graph} & \multirow{2}{*}{$n\log^{11} n\lambda_{s+2}(\log n)$}   & \multirow{2}{*}{$n\log^{9} n\lambda_{s+2}(\log n)$} \\
		&  & \\
		\hline
		\multirow{2}{5cm}{Dynamic connectivity in disk intersection graphs}  & $\Psi^2\log^9 n\lambda_{s+2}(\log n)$ update,  & $\Psi^2\log^7 n\lambda_{s+2}(\log n)$ update\\
		& $\log n/ \log\log n$ query & $\log n/ \log\log n$ query \\
		\hline 
		\multirow{2}{5cm}{BFS tree in a disk intersection graph} & \multirow{2}{*}{$n\log^{9} n\lambda_{s+2}(\log n)$}    &\multirow{2}{*}{$n\log^{7} n\lambda_{s+2}(\log n)$} \\
		&  & \\
		\hline
		\multirow{2}{5cm}{$(1+\rho)$-spanner for a disk intersection graph} & \multirow{2}{*}{$(n/\rho^2)\log^{9} n\lambda_{s+2}(\log n)$}    & \multirow{2}{*}{$(n/\rho^2)\log^{7} n\lambda_{s+2}(\log n)$} \\
		&  & \\
		\hline
	\end{tabular}
	\caption{Comparison on problems with respect to disk intersection graph.}\label{tb-disk}
	\renewcommand{\arraystretch}{1}
\end{table}
  
This paper is organized as follows.
Section~\ref{sec-random-sampling} introduces the configuration space and derives the random sampling technique.
Section~\ref{sec-shallow-cutting} formulates distance functions, designs the $\frac{1}{r}$-shallow-cutting, and proves its size to be $O(r)$. 
Section~\ref{sec-DS} composes the data structure for the $k$ nearest neighbors problem.
Section~\ref{sec-algorithm} presents the construction algorithm for shallow cuttings.
Section~\ref{sec-conclusion} makes concluding remarks.
Throughout the paper, if not explicitly stated, the base of the logarithm is 2.

\section{Random Sampling}\label{sec-random-sampling}
We first introduce the \emph{configuration space} and discuss several classical random sampling techniques. 
Then, we propose a new random sampling technique that utilizes relatively \emph{many} local conflicts to prevent relatively \emph{few} global conflicts,
in contradiction to most state-of-the-art works that adopt relatively \emph{few} local conflicts to prevent relatively \emph{many} global conflicts.
Finally, since our new technique requires some conditions,
we further prove that those conditions are sufficient at high probability. 
Our random sampling technique is very general, and for further applications,
we describe it in an abstract form.

\subsection{Configuration Space}\label{sub-rs-configuration}

Let $S$ be a set of $n$ objects, and for a subset $S'\subseteq S$, let $C(S')$ be the set of ``\emph{configurations}'' in a geometric structure defined by $S'$.  
For example, objects are planes in three dimensions, and a configuration in $C(S')$ is a tetrahedron in the so-called \emph{canonical triangulation}~\cite{AgarwalBMS98,Mulmuley94} for the \emph{arrangement} of the planes in $S'$, where the arrangement of planes partitions $\R^3$ into cells that intersect no plane and the canonical triangulation further partitions each cell into tetrahedra sharing the same bottom vertex.
Let $\T(S')$ be the set of all possible configurations defined by objects in $S'$, i.e., $\T(S')=\bigcup_{S''\subseteq S'}C(S'')$, and let $\T$ be $T(S)$.
In the above example, $|C(S')|=\Theta(|S'|^3)$~\cite{AgarwalBMS98,Mulmuley94}, while $|T(S')|=O(|S'|^{12})$ (since a tetrahedron has 4 vertices and a vertex is defined by 3 planes).

For each configuration $\triangle\in T$, we associate $\triangle$ with two subsets $D(\triangle), K(\triangle)\subseteq S$.
$D(\triangle)$, called the \emph{defining set}, defines $\triangle$ in a suitable geometric sense.
For instance, $\triangle$ is a tetrahedron, and $D(\triangle)$ is the set of planes that define the vertices of $\triangle$.
Let $d(\triangle)$ be $|D(\triangle)|$, and assume that for every $\triangle\in T$, $d(\triangle)\leq d$ for a small constant $d$.

$K(\triangle)$, called the \emph{conflict set}, comprises objects being said to \emph{conflict with} $\triangle$; let $w(\triangle)=|K(\triangle)|$.
The meaning of $K(\triangle)$  depends on the subject.
If $\triangle$ is a tetrahedron, for computing the arrangement of planes (\cite[Chapter 6]{Mulmuley94}), $K(\triangle)$ is the set of planes intersecting $\triangle$, while in our analysis for the expected number of relevant tetrahedra (Section~\ref{sub-ind-contribution}), $K(\triangle)$ is the set of planes lying fully below $\triangle$. 
In the latter example, $D(\triangle)$ may not be disjoint from $K(\triangle)$.

Furthermore, let $C^j(S')$ be the set of configurations $\triangle\in C(S')$ with $|K(\triangle)\cap S'|=j$ (i.e., the \emph{local} conflict size is $j$),
let $C_m(S')$ be the set of configurations $\triangle\in C(S')$ with $|K(\triangle)|=m$ (i.e., the \emph{global} conflict size is $m$),
and let $C_m^j(S')$ be $C^j(S')\cap C_m(S')$.

Most works focus on $C^0(S')$. Chazelle and Friedman~\cite{ChazelleF90} studied an if-and-only-if condition:
\begin{itemize}
	\item[($\ast$)] $\triangle\in C^0(S')$ if and only if $D(\triangle)\subseteq S'$ and $K(\triangle)\cap S'=\emptyset$.
\end{itemize}
Agarwal et~al.~\cite{AgarwalMS98} further relaxed the above condition by two weaker conditions:
\begin{enumerate}[label=(\roman*)]
	\item For any $\triangle\in C^0(S')$, $D(\triangle)\subseteq S'$ and $K(\triangle)\cap S'=\emptyset$.
	\item If $\triangle\in C^0(S')$ and $S''\subseteq S'$ with $D(\triangle)\subseteq S''$, then $\triangle\in C^0(S'')$.
\end{enumerate}
Condition~($\ast$) is both necessary and sufficient, while Condition~(i) is only necessary and Condition~(ii) is rather monotone.
For example, if a conflict relation is defined as that a plane intersects a tetrahedron, by Condition~($\ast$), $C^0(S')$ would be all tetrahedra in the canonical triangulation defined by $S'$, while by Conditions~(i)--(ii), $C^0(S')$ can be merely the tetrahedra in the canonical triangulation defined by $S'$ that intersect the $(\leq \frac{n}{r})$-level of the $n$ planes in $S$, i.e., $C^0(S')$ can be the \emph{relevant} tetrahedra in the canonical triangulation defined by $S'$.
(The $(\leq \frac{n}{r})$-level and relevant tetrahedra are already defined in Section~\ref{sec-ind}.)

\begin{remark}
	The property ``$K(\triangle)\cap S'=\emptyset$'' in Condition~($\ast$) and Condition~(i) is actually redundant in our notation.
	In contrast, Agarwal et~al.'s notation~\cite{AgarwalMS98} requires this property since they do not consider nonzero local conflicts. We keep this redundant property for easier comparison.
\end{remark}

Agarwal et~al.\ generalized Chazelle and Friedman's concept to bound the expected number of configurations that conflict with at least $t\frac{n}{r}$ objects, but no object in an $r$-element sample:

\begin{lemma}\label{lem-agarwal1} (\cite[Lemma~2.2]{AgarwalMS98})
For an $r$-element random subset $R$ of $S$, if $C^0(R)$ satisfies Conditions~(i)~and~(ii), then
\begin{equation*}
E[|C^0_{\geq t\frac{n}{r}}(R)|]=O(2^{-t})\cdot E[|C^0(R')|],
\end{equation*}
where $t$ is a parameter with $1\leq t\leq \frac{r}{d}$ and $R'$ is a random subset of $R$ of size $r'=\lfloor \frac{r}{t}\rfloor$.
\end{lemma}

In addition to the expected results, several high probability results exist if $\T(S)$ satisfies a property called \emph{bounded valence}: for all subsets $S'\subseteq S$, $|T(S')|=O(|S'|^d)$, and for all configurations $\triangle\in T(S')$, $D(\triangle)\subseteq S'$.

\begin{lemma}\label{lem-high-probability} (\cite[Theorem~5.1.2]{Mulmuley94})
	If $\T(S)$ satisfies the bounded valence, for an $r$-element random subset $R$ of $S$ and a constant $c>d$, with probability $1-O(r^{-(c-d)})$, every configuration in $C^{\bm{0}}(R)$ conflicts with at most $c\frac{n}{r}\log r$ objects in $S$.
\end{lemma}

The following corollary can be derived in a similar way as Lemma~\ref{lem-high-probability}.

\begin{corollary}\label{cor-high-probability}
If $\T(S)$ satisfies the bounded valence, for a random subset $R$ of $S$ of size $c\cdot t\log t$ and a sufficiently large constant $c$, with probability greater than $1/2$, every configuration in $C^{\bm{0}}(R)$ conflicts with at most $\frac{n}{t}$ objects in $S$.
\end{corollary}

\begin{table}
	\caption{Symbol Table.}\label{tb-symbols}
	\centering
	\begin{tabular}{|c|l|}
		\hline
		$C(R)$ & Configurations defined by $R$\\
		\hline
		$C^j(R)$ & Configurations in $C(R)$ in conflict with $j$ objects in $R$\\
		\hline  
		$C_m (R)$ & Configurations in $C(R)$ in conflict with $m$ objects in $S$\\
		\hline
		$C^j_m (R)$ & Configurations in $C(R)$ in conflict with $j$ objects in $R$ and $m$ objects in $S$\\
		\hline
		$\T$ & All possible configurations defined by objects in $S$, i.e., $T=\bigcup_{S'\subseteq S}C(S)$\\
		\hline
		$\T_m$ & Configurations in $\T$ in conflict with \textbf{at most} $m$ objects in $S$\\
		\hline
		$D(\triangle)$ & objects that define $\triangle$\\
		\hline
		$d(\triangle)$ & size of $D(\triangle)$\\
		\hline
		$K(\triangle)$ & objects that conflict with $\triangle$\\
		\hline
		$w(\triangle)$ & size of $K(\triangle)$\\
		\hline
		$x(\triangle)$ & $|K(\triangle)\cap D(\triangle)|$\\
		\hline
	\end{tabular}
\end{table}

Table~\ref{tb-symbols} illustrates symbols that we will use, and in the proofs of the following two subsections, $a^{\underline{b}}$ denotes $\prod_{i=0}^{b-1}(a-i)$.

\subsection{Many Local Conflicts Prevent Few Global Conflicts}\label{sub-rs-local-global}

Applications, e.g. the analysis for the expected number of relevant tetrahedra in Section~\ref{sub-ind-contribution} and our shallow cuttings in Section~\ref{sec-shallow-cutting}, would need to utilize relatively \emph{many} local conflicts to prevent relatively \emph{few} global conflicts.
Such utilization, however, is not allowed by Conditions~(i) and (ii) since they do not consider \emph{nonzero} local conflicts.

To include \emph{nonzero} local conflicts,  we generalize Conditions~(i) and (ii) with $t$ and $t'$ as follows:
\begin{enumerate}[label=(\Roman*)]
	\item For any $\triangle\in C^t(S')$, $D(\triangle)\subseteq S'$ and $|K(\triangle)\cap S'|=t$.
	\item If $\triangle\in C^t(S')$ and $S''\subseteq S'$ with $D(\triangle)\subseteq S''$ and $|K(\triangle)\cap S''|=t'$, then $\triangle\subseteq C^{t'}(S'')$.
\end{enumerate}

We establish Theorem~\ref{thm-many-local-less-global}, which roughly states that if 
the local conflict size of a random configuration is $t$, 
the probability that its global conflict size is linear in $\frac{n}{r}$ decreases factorially in $t$. 
Moreover, it is notable that Lemma~\ref{lem-agarwal1} has a \emph{lower} bound on the global conflict size,
while Theorem~\ref{thm-many-local-less-global} has an \emph{upper} bound.
As a result,
although our proof looks similar to that of Lemma~\ref{lem-agarwal1} (\cite[Lemma~2.2]{AgarwalMS98}), the derivation for bounding the probability is different,
and our proof chooses a sample size between $r-t+d$ and $r-t$ instead of $\lfloor\frac{r}{t}\rfloor$.

\begin{theorem}\label{thm-many-local-less-global}
Let $R$ be an $r$-element random subset of $S$ with $2d\leq r\leq\frac{n}{2}$,
and let $t$ be an integer with $d\leq t\leq r-d$.
If $C(R)$ satisfies Conditions~$(I)$ and $(II)$, then
\[E[|C^t_{\leq c\frac{n}{r}}(R)|]\leq \sum_{l=0}^d \frac{e^{2c}\cdot c^{t-l}}{(t-l)!}\cdot E[|C^l(R_l)|],\] 
where $R_l$ is an $(r-t+l)$-element random subset of $R$. ($R_l$ is also a random subset of $S$.)\jr{$l\longrightarrow i$}
\end{theorem}
\begin{proof}
To simplify descriptions, let $m$ be $c\frac{n}{r}$ and let $\T_m$ be the set of configurations in $\T$ that conflict with \emph{at most} $m$ objects in $S$. 
Consider a configuration $\triangle\in \T_m$, and let $x(\triangle)$ be $|D(\triangle)\cap K(\triangle)|$. It is clear that $0\leq x(\triangle)\leq d$.
We attempt to prove 

\begin{equation}\label{eq-x-2}
\Pr[\triangle\in C^t(R)]\leq \frac{e^{2c}\cdot c^{t-x(\triangle)}}{(t-x(\triangle))!}\cdot \Pr[\triangle\in C^{x(\triangle)} (R_{x(\triangle)})], 
\end{equation}
where $R_{x(\triangle)}$ is an $(r-t+x(\triangle))$-element random subset of $R$.
Then, we have
\[E[C^t_{\leq m}(R)]=\sum_{\triangle'\in \T_m} \Pr[\triangle'\in C^t(R)]\overset{(\ref{eq-x-2})}{\leq}  \sum_{\triangle'\in \T_m}  \frac{e^{2c}\cdot c^{t-x(\triangle')}}{(t-x(\triangle'))!}\Pr[\triangle'\in C^{x(\triangle')} (R_{x(\triangle')})]\]
\[ \leq  \sum_{\triangle'\in \T_m}  \sum_{l=0}^d \frac{e^{2c}\cdot c^{t-l}}{(t-l)!} \Pr[\triangle'\in C^l (R_l)] \leq \sum_{\triangle''\in \bm{T}}  \sum_{l=0}^d \frac{e^{2c}\cdot c^{t-l}}{(t-l)!} \Pr[\triangle''\in C^l (R_l)]\]
\[=   \sum_{l=0}^d \frac{e^{2c}\cdot c^{t-l}}{(t-l)!} \sum_{\triangle''\in T} \Pr[\triangle''\in C^l (R_l)]=\sum_{l=0}^d \frac{e^{2c}\cdot c^{t-l}}{(t-l)!}\cdot E[|C^l(R_l)|].\]

Let $\triangle\in \T_m$ be a \emph{\textbf{fixed}} configuration. Also let us assume that $t\leq w(\triangle)$ holds; otherwise, $\triangle$ must not belong to $C^t(R)$, making the claim~(\ref{eq-x-2}) obvious.

Let $A_\triangle$ be the event that $D(\triangle)\subseteq R$ and $|K(\triangle)\cap R|=t$,
and let $A'_\triangle$ be the event that $D(\triangle)\subseteq R_{x(\triangle)}$ and $K(\triangle)\cap R_{x(\triangle)}=K(\triangle)\cap D(\triangle)$,
the latter of which implies that $|K(\triangle)\cap R_{x(\triangle)}|=x(\triangle)$.

According to Condition~(I), we have
\[ \Pr[\triangle\in C^t(R)]=\Pr[A_\triangle]\cdot \Pr[\triangle\in C^t(R)\mid A_\triangle],\]
and
\[ \Pr[\triangle\in C^{x(\triangle)}(R_{x(\triangle)})]=\Pr[A'_\triangle]\cdot \Pr[\triangle\in C^{x(\triangle)}(R_{x(\triangle)})\mid A'_\triangle].\]
So, we have
\[\frac{\Pr[\triangle\in C^t(R)]}{\Pr[\triangle\in C^{x(\triangle)}(R_{x(\triangle)})]}=\frac{\Pr[A_\triangle]}{\Pr[A'_\triangle]}\cdot \frac{\Pr[\triangle\in C^t(R)\mid A_\triangle]}{\Pr[\triangle\in C^{x(\triangle)}(R_{x(\triangle)})\mid A'_\triangle]}.\]

Moreover, according to Condition~(II), we have
\[\Pr[\triangle\in C^t(R)\mid A_\triangle]\leq \Pr[\triangle\in C^{x(\triangle)}(R_{x(\triangle)})\mid A'_\triangle],\]
implying that
\begin{equation}\label{eq-y-2}
\frac{\Pr[\triangle\in C^t(R)]}{\Pr[\triangle\in C^{x(\triangle)}(R_{x(\triangle)})]}\leq \frac{Pr[A_\triangle]}{Pr[A'_\triangle]}.
\end{equation}
(The implication from Condition~(II) can be explained through a random experiment similar to the random experiment in the proof of~\cite[Lemma~2.2]{AgarwalMS98}; see the end of the current proof.)

Let $\ell$ be $x(\triangle)$,  $w$ be $w(\triangle)$ and $r'$ be $|R_{x(\triangle)}|$, i.e., $r'=r-(t-x(\triangle))=r-(t-\ell)$. Recall that $w\leq m=c\frac{n}{r}$ since $\triangle\in \T_m$.

\[
\def\arraystretch{3}
\begin{array}{>{\displaystyle}l}
\frac{Pr[A_\triangle]}{Pr[A'_\triangle]}=\frac{\frac{{d(\triangle)\choose d(\triangle)}{w-\ell\choose t-\ell}{n-d(\triangle)-(w-\ell)\choose r-d(\triangle)-(t-\ell)}}{{n\choose r}}}{\frac{{d(\triangle)\choose d(\triangle)}{n-d(\triangle)-(w-\ell)\choose r'-d(\triangle)}}{{n\choose r'}}}=\frac{\frac{{w-\ell\choose t-\ell}{n-d(\triangle)-(w-\ell)\choose r-d(\triangle)-(t-\ell)}}{{n\choose r}}}{\frac{{n-d(\triangle)-(w-\ell)\choose r-d(\triangle)-(t-\ell)}}{{n\choose r-(t-\ell)}}}=\frac{{n\choose r-(t-\ell)}{w-\ell\choose t-\ell}}{{n\choose r}}\\
=\frac{n!}{(r-(t-\ell))!(n-r+(t-\ell))!}\cdot \frac{r!(n-r)!}{n!}\cdot\frac{(w-\ell)!}{(t-\ell)!(w-t)!}\\
=\underbrace{\frac{r^{\underline{t-\ell}}}{(n-r+(t-\ell))^{\underline{t-\ell}}}\cdot\frac{(w-\ell)^{\underline{t-\ell}}}{(t-\ell)!}\leq \frac{r^{t-\ell}\cdot (w-\ell)^{t-\ell}}{(n-r)^{t-\ell}}\cdot \frac{1}{(t-\ell)!}}_{r^{\underline{t-\ell}}\leq r^{t-\ell},\; (w-\ell)^{\underline{t-\ell}}\leq (w-\ell)^{t-\ell},\; (n-r+(t-\ell))^{\underline{t-\ell}}\geq (n-r)^{t-\ell}}\\
\leq \frac{r^{t-\ell}\cdot (c\cdot \frac{n}{r})^{t-\ell}}{(n-r)^{t-\ell}}\cdot \frac{1}{(t-\ell)!}=\frac{c^{t-\ell}}{(t-\ell)!}\cdot (\frac{n}{n-r})^{t-\ell}=\frac{c^{t-\ell}}{(t-\ell)!}\cdot (1+\frac{r}{n-r})^{t-\ell}\\
\leq \frac{c^{t-\ell}}{(t-\ell)!}\cdot e^{\frac{r(t-\ell)}{n-r}}\leq  \frac{c^{t-\ell}}{(t-\ell)!}\cdot e^{\frac{cn}{n-r}}\leq \frac{e^{2c}\cdot c^{t-\ell}}{(t-\ell)!},\\
\end{array}
\]
which derives the claim~(\ref{eq-x-2}) from the claim~(\ref{eq-y-2}).
The second to last inequality comes from the fact that $r(t-\ell)\leq rw\leq r\cdot c\frac{n}{r}$,
and the last inequality comes from the fact that $\frac{n}{n-r}\leq 2$ (since $r\leq \frac{n}{2}$).

Since $R$ needs to contain all the objects in $D(\triangle)$ and exactly $t$ objects in $K(\triangle)$,
we let $r$ be at least $t+d$ to allow the case that $|D(\triangle)|=d$ and $D(\triangle)\cap K(\triangle)=\emptyset$;
we also let $t$ be at least $d$ to allow the case that $D(\triangle)\subseteq K(\triangle)$.
These two settings lead to the condition that $d\leq t\leq r-d$.

A random experiment to explain the implication from Condition~(II) is stated as follow: first select a set $R_{x(\triangle)}$ by including all the $d(\triangle)$ objects of $D(\triangle)$ into $R_{x(\triangle)}$, and adding $r-d(\triangle)-t+x(\triangle)$ randomly chosen objects of $S\setminus \big(D(\triangle)\cup K(\triangle)\big)$. By definition of the conditional probability, the probability that $\triangle\in C^{x(\triangle)}(R_{x(\triangle)})$ is exactly $\Pr[\triangle\in C^{x(\triangle)}(R_{x(\triangle)})\mid A'_\triangle]$. Then take this $R_{x(\triangle)}$ and add $t-x(\triangle)$ randomly chosen objects of $K(\triangle)\setminus D(\triangle)$ to $R_{x(\triangle)}$, obtaining a set $R$. It is clear that the distribution of $R$ is the same as if we took the $d(\triangle)$ objects of $D(\triangle)$ , added $t-x(\triangle)$ randomly chosen objects of $K(\triangle)\setminus D(\triangle)$, and added $r-d(\triangle)-t+x(\triangle)$ randomly chosen objects of $S\setminus \big(D(\triangle)\cup K(\triangle)\big)$. Hence, the probability that $\triangle\in C^t(R)$ is $\Pr[\triangle\in C^t(R)\mid A_\triangle]$. Since $R$ was created by adding extra objects to $R_{x(\triangle)}$, i.e., $R_{x(\triangle)}\subseteq R$, Condition~(II) implies that whenever $\triangle$ appears in $C^t(R)$, it must appear in $C^{x(\triangle)}(R_{x(\triangle)})$, leading to that $\Pr[\triangle\in C^t(R)\mid A_\triangle]\leq \Pr[\triangle\in C^{x(\triangle)}(R_{x(\triangle)})\mid A'_\triangle]$.
\end{proof}

\begin{remark}\label{rm-clarkson}
	Clarkson also derived a factorial bound that $E[|T^t_{\leq \frac{n}{r}}(R)|]\leq O\big((\frac{e}{t})^t\big)\cdot |T(R)|$,
	where $T^t_{\leq \frac{n}{r}}(R)$ is the set of configuration in $T(R)$ that conflict with $t$ objects in $R$ and at most $\frac{n}{r}$ objects in $S$ (\cite[Corollary~4.3]{Clarkson87}). However, since there is a quantity difference between $|C^{\leq d}(R)|$ and $|T(R)|$, his factorial bound could not be applied to derive linear-size shallow cuttings. For example, regarding the canonical triangulation for planes in $\R^3$, $|C^{\leq d}(R)|=O(r)$, but $|T(R)|=O(r^{12})$. A detailed comparison is in Appendix~\ref{ap-clarkson}.
\end{remark}

\begin{remark}
	One could assume $D(\triangle)\cap K(\triangle)$ to be empty by replacing $K(\triangle)$ with $K(\triangle)\setminus D(\triangle)$, but the main issue is that when $|K(\triangle)\cap S'|>|D(\triangle)\cap K(\triangle)|$, even after the replacement, $K(\triangle)\cap S'$ is still not empty, distinguishing the technical details of bounding $\frac{Pr[A_\triangle]}{Pr[A'_\triangle]}$ in the proof of Theorem~\ref{thm-many-local-less-global} from previous works~\cite{AgarwalMS98,ChazelleF90,deBergDS1995,Matousek92a}.\jr{Say the proof actually follows a reverse direction!!!}
\end{remark}

\subsection{Logarithmic Local Conflicts are Enough}\label{sub-rs-local-log}
Theorem~\ref{thm-many-local-less-global} requires $t$ to be at most $r-d$,
but $t$ could be $r$ in the worst case. Therefore, we will prove that at high probability, $t$ is $O(\log r)$.

First of all, we analyze the probability that a configuration conflicts with few elements in $S$, but \emph{relatively} many elements in $R$.

\begin{lemma}\label{lem-high-level}
	Let $R$ be an $r$-element random subset of $S$, let $\triangle$ be a configuration with $D(\triangle)\subseteq R$,
	let $w$ be $|K(\triangle)|$, and let $t$ be $|K(\triangle)\cap R|$.
	If $\max\{c(d+m)\log r+d+\max\{c,d\}, 2^{2^{1/c}\cdot\frac{e}{(d+m)}}\}\leq r\leq \frac{n}{(d+m)\cdot \log n}$ and $w\leq c\cdot\frac{n}{r}$,
	then the probability that $t\geq c(d+m)\log r+d$ is at most $r^{-(d+m)}$.  
\end{lemma}
\begin{proof}
	Let $d'$ be $|D(\triangle)|$, and let $\ell$ be $|D(\triangle)\cap K(\triangle)|$. It is clear that $\ell\leq d' \leq d$.
	Since $R$ must contain all the $d'$ elements of $D(\triangle)$ and exactly $t-\ell$ elements of $K(\triangle)\setminus D(\triangle)$,
	the probability is 
	\[\frac{{d'\choose d'}{w-\ell\choose t-\ell}{n-(w-\ell)-d'\choose r-d'-(t-\ell)}}{{n\choose r}}=\frac{(w-\ell)!}{(t-\ell)!(w-t)!}\frac{(n-(w-\ell)-d')!}{(r-d'-(t-\ell))!(n-w-r+t)!}\frac{r!(n-r)!}{n!}\]
	\[=\frac{(w-\ell)^{\underline{t-\ell}}}{(t-\ell)!}\frac{(n-r)^{\underline{w-t}}\cdot r^{\underline{d'+(t-\ell)}}}{n^{\underline{(w-\ell)+d'}}}=\frac{1}{(t-\ell)!}\cdot\frac{(w-\ell)^{\underline{t-\ell}}r^{\underline{t-\ell}}}{n^{\underline{t-\ell}}}\cdot \frac{(n-r)^{\underline{w-t}}(r-(t-\ell))^{\underline{d'}}}{(n-(t-\ell))^{\underline{w+d'-t}}}\]
	\[=\frac{1}{(t-\ell)!}\cdot\frac{(w-\ell)^{\underline{t-\ell}}r^{\underline{t-\ell}}}{n^{\underline{t-\ell}}}\cdot \underbrace{\frac{(r-(t-\ell))^{\underline{d'}}}{(n-(t-\ell))^{\underline{d'}}}\cdot \frac{(n-r)^{\underline{w-t}}}{(n-(t-\ell)-d')^{\underline{w-t}}}}_{r-(t-\ell)\leq n-(t-\ell) \mbox{ and } n-r \leq n-(t-\ell)-d'\; (\mbox{since } r\geq (t-\ell)+d')}\]
	\[\leq \frac{1}{(t-\ell)!}\cdot\frac{(w-\ell)^{\underline{t-\ell}}r^{\underline{t-\ell}}}{n^{\underline{t-\ell}}}\leq \frac{1}{(t-\ell)!}\cdot \frac{(c\frac{n}{r}-\ell)^{\underline{t-\ell}}r^{\underline{t-\ell}}}{n^{\underline{t-\ell}}}=\frac{1}{(t-\ell)!}\cdot \prod_{i=0}^{t-\ell-1}\frac{(c\frac{n}{r}-i-\ell)(r-i)}{n-i}.\]
	Since $c(n-i)-(c\frac{n}{r}-i-\ell)(r-i)=(i+\ell)(r-i)+ci(\frac{n}{r}-1)>0$, we have $\frac{(c\frac{n}{r}-l-i)(r-i)}{n-i}\leq c$, implying that the probability is at most
	\[\frac{c^{t-\ell}}{(t-\ell)!}\leq (\frac{c\cdot e}{t-\ell})^{t-\ell}\leq (\frac{c\cdot e}{c\cdot(d+m)\log r})^{c\cdot(d+m)\log r}\leq (\frac{1}{2})^{(d+m)\log r}=r^{-(d+m)}.\]
	The first inequality comes from Stirling's approximation,
	the second inequality comes from the fact that $(\frac{1}{(t-\ell)})^{(t-\ell)}$ is inversely proportional to $(t-\ell)$ and that $t-\ell\geq c\cdot(d+m)\log r+d-\ell\geq c\cdot(d+m)\log r$,
	and the third inequality comes from the fact that $r\geq 2^{2^{1/c}\cdot\frac{e}{(d+m)}}$, i.e., $(\frac{e}{(d+m)\log r})^c\leq 1/2$.
\end{proof}

Then, we assume that $\T(S)$ satisfies the bounded valence, i.e., for any subset $S'\subseteq S$, 
$|T(S')|=O(|S'|^d)$ and for any configuration $\triangle\in T(S')$, $D(\triangle)\subseteq S'$, and prove the following theorem.

\begin{theorem}\label{thm-log-level}
	Let $R$ be an $r$-element random subset of $S$.
	If $\T(S)$ satisfies the bounded valence and $\max\{c(d+m)\log r +d+\max\{c, d\}, 2^{2^{1/c}\cdot\frac{e}{(d+m)}}\}\leq r\leq \frac{n}{(d+m)\cdot \log n}$,
	then the probability that there exists a configuration in $C(R)$ in conflict with at most $c\frac{n}{r}$ objects in $S$, but at least $c(d+m)\log r+d$ objects in $R$ is $O(r^{-m})$.
\end{theorem}
\begin{proof}
	Since $\T(S)$ satisfies the bounded valence, $|C(R)|\leq |T(R)|=O(r^d)$. 
	By Lemma~\ref{lem-high-level} and the union bound,
	the probability is $O(r^d)\cdot r^{-(d+m)}=O(r^{-m})$. 
\end{proof}

\begin{remark}
	Corollary 4.3 by Clarkson~\cite{Clarkson87} can also lead to the same result, but he adopted the random sampling with replacement. 
	Since $r$ is $\Omega(\frac{n}{\log n})$ in our situation, his random sampling gets a multi-set with high probability, and thus his result could not be directly applied. (In his applications, either $r$ is far from $n$ or a multi-set is feasible.) 
\end{remark}

\section{Shallow Cutting}\label{sec-shallow-cutting}

We first formulate the distance functions as surfaces and introduce the vertical decomposition of surfaces.
Then, we design a \emph{$\frac{1}{r}$-shallow-cutting} for $n$ surfaces using vertical decompositions.\jr{not informative enough.}
Finally, we adopt our random sampling technique to prove the expected size of our $\frac{1}{r}$-shallow-cutting to be $O(r)$.

\subsection{Distance Functions and Vertical Decomposition}\label{sub-sc-distance-VD}

Let $S$ be a set of $n$ pairwise disjoint sites that are simply-shaped compact convex regions of constant description complexity in the plane (e.g., line segments, disks, squares) and let $\tau$ be a continuous distance function between two points in the plane.
Assume that $\tau$ and the sites in $S$ are defined by a constant number of polynomial equations and inequalities of constant maximum degree. 
For each site $q\in S$, define its distance function $f_q$ with respect to any point $(x,y)\in \R^2$ as $f_q(x,y)=\min\{\tau\big((x,y),p\big)\mid p\in q\}$, and let $F$ denote the set of distance functions $\{f_q\mid q\in S\}$.

The graph of each function in $F$ is a semialgebraic set, defined by a constant number of polynomial equations and inequalities of constant maximum degree. The lower envelope $\E_F$ of $F$ is the graph of the pointwise minimum of the functions in $F$; the upper envelope is defined symmetrically.
We assume that for any subset $R\subseteq F$, the lower envelope $\E_R$ has $O(|R|)$ faces, edges, and vertices.
This assumption holds for many distance functions, e.g., for point sites in any constant-size algebraic convex distance metric and additively weighted Euclidean distances, 
and for disjoint line segments, disks, and constant-size convex polygons in the $L_p$ norms or under the Hausdorff metric.

For conceptual simplicity, each function in $F$ is represented as an $xy$-monotone surface in $\R^3$.
A general position assumption is made on $F$: no more than three surfaces intersect at a common point, no more than two surfaces intersect in a one-dimensional curve, no pair of surfaces are tangent to each other, and if two surfaces intersect, their intersection are one-dimensional curves.
Moreover, $\bm{s}$ is defined as the maximum number of co-vertical pairs of points $q, q'$ with $q\in f\cap g$, $q'\in f'\cap g'$ over all quadruples $f,g,f',g'$ of \emph{distinct} surfaces in $F$, and $s$ is assumed to be a constant.
For a point $p\in\R^3$, the \emph{level} of $p$ with respect to $F$ is the number of surfaces in $F$ lying below $p$, and the $(\leq l)$-level of $F$ is the set of points in $\R^3$ whose level with respect to $F$ is at most $l$.

For a subset $R\subseteq F$, let $\A(R)$ be the \emph{arrangement} formed by the surfaces in $R$. 
For each cell $C$ in $\A(R)$, its boundary consists of a \emph{ceiling} and a \emph{floor}. 
The \emph{ceiling} is a part of the \emph{lower} envelope of surfaces in $R$ that lie \emph{above} $C$, and the \emph{floor} is a part of the \emph{upper} envelope of surfaces in $R$ that lie \emph{below} $C$.
The topmost (resp.\ bottommost) cell in $\A(R)$ does not have a ceiling (resp.\ a floor). 
If the level of $C$ with respect to $R$ is $t$, then the vertical line through a point in $C$ intersects the boundary of $C$ at its $(t+1)^\mathrm{st}$ and $t^\mathrm{th}$ lowest surfaces in $R$.

The \emph{vertical decomposition} $\VD(R)$ of $R$, proposed by Chazelle et~al.~\cite{ChazelleEGS91}, decomposes each cell $C$ of $\A(R)$ into \emph{pseudo-prisms} or shortly \emph{prisms}, a notion to be defined below; we also refer to \cite[Section~8.3]{SA95}.
First, we project the ceiling and the floor of $C$, namely their edges and vertices, onto the $xy$-plane, and overlap the two projections.  
Then, we build the so-called \emph{vertical trapezoidal decomposition}~\cite{deBergCKO08,Mulmuley94} for the overlap between the two projections by erecting a $y$-vertical segment from each vertex, from each intersection point between edges, and from each locally $x$-extreme point on the edges, which yields a collection of pseudo-trapezoids.  
Finally, we extend each pseudo-trapezoid $\triangle$ to a trapezoidal prism $\triangle\times \R$, and form a prism $\Diamond=(\triangle\times\R)\cap C$. 

Each prism has six faces, top, bottom, left, right, front, and back.
Its top (resp.\ bottom) face is a part of a surface in $F$.
Its left (resp.\ right) face is a part of a plane perpendicular to the $x$-axis. 
Its front (resp.\ back) face is a part of a vertical wall through a intersection curve between two surfaces in $F$. 
Its top and bottom faces are kind of \emph{pseudo-trapezoids} on their respective surfaces, so a prism is also the collection of points lying vertically between the two pseudo-trapezoids.

$\VD(R)$ contains $O(|R|^4)$ prisms and can be computed in $O(|R|^5\log|R|)$ time~\cite{ChazelleEGS91}. The prisms in the topmost and bottommost cells of $\A(R)$ are semi-unbounded. For our algorithmic aspect, we imagine a surface $f_\infty: z=\infty$, so that each prism in the topmost cell has a top face lying on $f_\infty$.
For each prism $\Diamond\in \VD(R)$, let $F_\Diamond$ denote the set of surfaces in $F$ intersecting $\Diamond$.

A prism is defined by at most 10 surfaces under the general position assumption. 
First, its top (resp.\ bottom) face belongs to a surface, and we call this surface top (resp.\ bottom).
Then, we look at the pseudo-trapezoid $\triangle$ that is the $xy$-projection of $\Diamond$.
A pseudo-trapezoid is defined by \emph{bisecting curves} in the plane, each of which is the $xy$-projections of an intersection curve between two surfaces. 
Since a bisecting curve defining $\triangle$ must be associated with the top surface or the bottom surface of $\Diamond$,
it is sufficient to bound the number of bisecting curves to define $\triangle$, and each of those bisecting curves counts one additional surface.
The upper (resp.\ lower) edge of $\triangle$ belongs to one bisecting curve. 
The left (resp.\ right) edge of $\triangle$ belongs to a vertical line passing through the left (resp.\ right)  endpoint of the upper edge, the left (resp.\ right) endpoint of the lower edge, or an $x$-extreme point of a bisecting curve.
Each of the first two cases results from one additional bisecting curve, namely each counts for one additional surface.
Although the last case may occur more than once, the same as Chazelle's algorithm~\cite{Chazelle91},
we can introduce zero-width trapezoids to solve such degenerate situation.

\subsection{Design of Shallow Cutting}\label{sub-sc-design}

A \emph{$\frac{1}{r}$-shallow-cutting} for $F$ is a set of \emph{disjoint} prisms satisfying the following three conditions:

\begin{enumerate}[label=(\alph*)]
	\item They cover the ($\leq\frac{n}{r}$)-level of $F$.
	\item Each of them is intersected by $O(\frac{n}{r})$ surfaces in $F$.
	\item They are downward semi-unbounded, i.e., no bottom face, so they do not vertically overlap .
\end{enumerate}

To design such a \emph{$\frac{1}{r}$-shallow-cutting}, we take an $r$-element \emph{random} subset $R$ of $F$ and adopt $R$ to generate prisms satisfying the three conditions.

For condition~(a), it is natural to consider the prisms in $\VD(R)$ that intersect the ($\leq\frac{n}{r}$)-level of $F$, but it is hard to compute those prisms exactly.
Thus, we instead select a super set $\AD(R)$ that consists of prisms in $\VD(R)$ lying fully above at most $\frac{n}{r}$ surfaces in $F$.

For condition~(b), if a prism $\Diamond\in\AD(R)$ intersects more than $\frac{n}{r}$ surfaces in $F$, we will refine it into smaller prisms, and select the ones lying fully above at most $\frac{n}{r}$ surfaces in $F$. 
This refinement is similar to a classical process proposed by Chazelle and Friedman~\cite{ChazelleF90}.
Let $t$ be $\lceil |F_\Diamond|/(\frac{n}{r})\rceil$, where
$F_\Diamond$ is the set of surfaces in $F$ intersecting $\Diamond$.
If $t>1$, we refine $\Diamond$ as follows:
\begin{enumerate}
	\item Take a random subset $F'$ of $F_\Diamond$ of size $O(t\log t)$, and construct $\VD(F')\cap \Diamond$ 
	by clipping each surface in $F'$ with $\Diamond$, building the vertical decomposition on the clipped surfaces plus 
	the top and bottom faces of $\Diamond$,
	and including the prisms lying inside $\Diamond$.
	\item If one prism in $\VD(F')\cap \Diamond$
	intersects more than $\frac{|F_\Diamond|}{t}$ surfaces in $F_\Diamond$, then repeat Step~1.
	\item For each prism $\Diamond'\in \VD(F')\cap \Diamond$, 
	if $\Diamond'$ lies fully above more than $\frac{n}{r}$ surfaces in $F$, 
	discard $\Diamond'$. 
\end{enumerate}
By defining a conflict between a surface and a prism as the surface intersects the prism,
Corollary~\ref{cor-high-probability} guarantees the existence of $F'$ that passes Step~2.
As stated in Section~\ref{sub-sc-distance-VD}, $\VD(F')\cap \Diamond$ has $O(t^4\log^4t)$ prisms.
For each prism in $\VD(F')\cap \Diamond$,
since $t=\lceil |F_\Diamond|/(\frac{n}{r}) \rceil$, it intersects at most $\frac{|F_\Diamond|}{t}\leq \frac{n}{r}$ surfaces in $F$,
and if it lies fully above more than $\frac{n}{r}$ surfaces in $F$, 
it will be discarded,
implying that 
each resulting prism intersects or lies fully above at most $2\frac{n}{r}$ surfaces in $F$.\jr{rephrasing}
Let $\RD(R)$ be the set of resulting prisms (including the unrefined prisms in $\AD(R)$).

For condition~(c),
we generate a set $\SC(R)$ of \emph{downward semi-unbounded} prisms from $\RD(R)$ in two steps. 
First, we build the upper envelope of the top faces of all prisms in $\RD(R)$.
Then, we decompose the region in $\R^3$ below the upper envelope into downward semi-unbounded prisms similarly to the decomposition of the bottommost cell in Section~\ref{sub-sc-distance-VD}, i.e., project the upper envelope onto th $xy$-plane, build the vertical trapezoidal decomposition for the projection, extend each trapezoid to a trapezoidal prism, and take the part of each prism below the upper envelope.
Since each prism in $\RD(R)$ intersects or lies fully above at most $2\frac{n}{r}$ surfaces in $F$, its top face intersects or lies fully above at most $2\frac{n}{r}$ surfaces in $F$,
and since the top face of each prism in $\SC(R)$
is a part of the top face of one prism in $\RD(R)$,
each prism in $\SC(R)$ intersects at most $2\frac{n}{r}$ surfaces in $F$.

\begin{remark}
	Actually, $\RD(R)$ is identical to Agarwal et~al.'s shallow cutting (\cite[Section~3]{AgarwalES99}),
	which satisfies the first two conditions, but not the third one.
	Their analysis for the expected size of $\RD(R)$ (\cite[Theorem~3.1]{AgarwalES99}) extends Matou\v{s}ek's analysis~\cite{Matousek92a} together with additional machinery for general distance functions,
	while our analysis (Theorem~\ref{thm-size-AD-RD}) makes use of our new random sampling technique (Theorem~\ref{thm-many-local-less-global}). 
	Especially, our description of generating $\RD(R)$ directly supports our analysis for the expected size of $\SC(R)$ (Theorem~\ref{thm-cutting-size}).
\end{remark}

\subsection{Structure Complexity}\label{sub-sc-size}

We will apply the random sampling techniques in Section~\ref{sec-random-sampling} to prove that the expected size of our shallow cutting in Section~\ref{sub-sc-design}, i.e., $E[|SC(R)|]$, is $O(r)$,
which confirms the existence of linear-size shallow cuttings for the $k$ nearest neighbors problem under general distance functions.
The high-level idea is that if a relevant prism (i.e., a prism in $\AD(R)$) lies fully above $\ell$ surfaces in $R$ (i.e., $\ell$ sample surfaces),
it contributes $O(1+\ell^4)$ to the value of $E[|SC(R)|]$, and the expected number of such relevant prisms is $O(\frac{r}{\ell!})$, implying a bound of $\sum_{\ell\geq 0}O(\frac{1+\ell^4}{\ell!}r)=O(r)$. 
Recall that the maximum number of surfaces to define a prism is 10 (Section~\ref{sub-sc-distance-VD}), and this number is represented by $d$ in the following proofs. 

As a warm-up, we first prove that both the expected sizes of $\AD(R)$ and $\RD(R)$ are $O(r)$, which we will apply to analyze the construction time in Section~\ref{sub-al-time}.

\begin{theorem}\label{thm-size-AD-RD}
If $117\leq r\leq \frac{n}{14\log n}$,  then $E[|\AD(R)|]=O(r)$ and $E[\RD(R)]=O(r)$.
\end{theorem}
\begin{proof}
	To analyze $E[|\AD(R)|]$, following the notations in Section~\ref{sec-random-sampling},
	an object is a surface, a configuration is a prism,
	and a surface $f$ is said to conflict with a prism $\Diamond$ if $f$ lies fully below $\Diamond$. 
	Let $C(R)$ be $\VD(R)$, so that $C_{\leq \frac{n}{r}}(R)$ is exactly $\AD(R)$, $C^t(R)$ is the set of prisms in $\VD(R)$ lying fully above $t$ surfaces in $R$, 
	and $C^t_{\leq \frac{n}{r}}(R)$ is the set of prisms in $\AD(R)$ lying above $t$ surfaces in $R$.
	Note that $d$, the maximum number of surfaces to define a prism, is 10.
	
	We first use Theorem~\ref{thm-log-level} to show that it is sufficient to consider the case in which all prisms in $\AD(R)$ have a level with respect to $R$ of at most $(d+4)\log r+d$, and then adopt Theorem~\ref{thm-many-local-less-global} to derive $E[|\AD(R)|]$.
	
	By setting $c=1$ and $m=4$,
	Theorem~\ref{thm-log-level} implies that the probability that there exists a prism in $\VD(R)$ lying fully above at most $\frac{n}{r}$ surfaces in $F$, but fully above at least $(d+4)\log r+d$ surfaces in $R$ is $O(r^{-4})$.
	In other words, the probability that $\AD(R)$ contains a prism whose level with respect to $R$ is at least  $(d+4)\log r+d$ is only $O(r^{-4})$.
	Since $|\AD(R)|\leq|\VD(R)|=O(r^4)$, the exception contributes only $O(r^{-4})\cdot O(r^4) = O(1)$ to $E[|\AD(R)|]$.

	Since $r\geq 117$ and $d=10$, we have $(d+4)\log r+d\leq r-d$ and thus have $[d,  (d+4)\log r+d]\subseteq [d, r-d]$.
	Therefore, Theorem~\ref{thm-many-local-less-global} implies that
	for $t\in [d,  (d+4)\log r+d]$,
	\begin{ceqn}
		\begin{align}\label{eq-cutting-inequality}
				E[|C^t_{\leq \frac{n}{r}}(R)|]\leq \sum_{l=0}^d \frac{e^2}{(t-l)!}\cdot E[|C^l(R_l)|],
		\end{align}
	\end{ceqn}
	where $R_l$ is an $(r-t+l)$-element random subset of $F$.
	By \cite[Lemma~5.1]{KaplanMRSS17}, we have $|C^t(R)|=O(r\cdot (t+1)\cdot \lambda_{s+2}(t+1))=O(r\cdot (t+1)^3)$,
	and by Inequality~(\ref{eq-cutting-inequality}), we have
	\begin{equation*} 
	\begin{split}
	E[|\AD(R)|] & =\sum_{t=0}^{d} \underbrace{E[|C^t_{\leq \frac{n}{r}}(R)|]}_{\leq E[|C^{t}(R)|]=O(r\cdot (t+1)^3)=O(r)}+ \sum_{t=d+1}^{(d+4)\log r+d} E[|C^t_{\leq \frac{n}{r}}(R)|] \\
	& \overset{(\ref{eq-cutting-inequality})}{\leq} O(r) + \sum_{t=d+1}^{(d+4)\log r+d}\sum_{l=0}^d O\big(\underbrace{\frac{e^2}{(t-l)!}}_{\leq \frac{e^2}{(t-d)!}}\big)\cdot \underbrace{E[|C^l(R_l)|]}_{O\big(r\cdot (l+1)^3\big)=O(r)}\\
	& = O(r)\cdot \underbrace{d\cdot e^2}_{O(1)}\cdot  \underbrace{\sum_{t=d+1}^{(d+4)\log r+d}\frac{1}{(t-d)!}}_{O(1)}=O(r)
	\end{split}
	\end{equation*}

	To analyze $E[|\RD(R)|]$, re-define a conflict relation as that a surface intersects a prism, and let  $C^0(R)$ be $\AD(R)$, so that $C^0_{\geq t\frac{n}{r}}(R)$ is the set of prisms in $\AD(R)$ that intersect at least $t\frac{n}{r}$ surfaces in $F$. 
	According to the generation of $\RD(R)$ in Section~\ref{sub-sc-design}, 
	a prism in $C^0_{\geq t\frac{n}{r}}(R) \setminus C^0_{\geq (t+1)\frac{n}{r}}(R)$ will be refined into $O\big((t+1)^4\log^4 (t+1)\big)=O\big((t+1)^5\big)$ ones. Moreover,
	by Lemma~\ref{lem-agarwal1}, 
	for $t\geq 1$, 
	$E[|C^0_{\geq t\frac{n}{r}}(R)|]=O(2^{-t})\cdot E[|C^0(R')|]=O(2^{-t})\cdot E[|\AD(R')|]=O(2^{-t}\cdot \frac{r}{t})=O(2^{-t}\cdot r)$, where $R'$ is a $\lfloor \frac{r}{t}\rfloor$-element random subset of $R$, concluding that
	\begin{equation*} 
	\begin{split}
	E[|\RD(R)|] & = E[|C^0(R)|+\sum_{t\geq 1}E[|C^0_{\geq t\frac{n}{r}}(R) \setminus C^0_{\geq (t+1)\frac{n}{r}}(R)|]\cdot O((t+1)^5) \\
	& \leq \sum_{t\geq 0} E[|C^0_{\geq t\frac{n}{r}}(R)|]\cdot O((t+1)^5)=\sum_{t\geq 0} O(\frac{(t+1)^5}{ 2^t }\cdot r)=O(r).
	\end{split}
	\end{equation*}
	
\end{proof}

To derive $E[|\SC(R)|]$, we first show that the expected number of prisms in $\AD(R)$ that lie fully above $\ell$ surfaces in $R$ and intersects at least $t\frac{n}{r}$ surfaces in $F$ is $O(\frac{r}{\ell!2^t})$, i.e., decrease factorially in $\ell$ and exponentially in $t$.
Formally,
let $\AD^\ell(R)$ be the set of prisms in $\AD(R)$ lying fully above $\ell$ surfaces in $R$, let $\AD_{\geq t\frac{n}{r}}$ be the set of prisms in $\AD(R)$ intersecting at least $t\frac{n}{r}$ in $F$, and let $\AD_{\geq t\frac{n}{r}}^\ell(R)$ be $\AD^\ell(R)\cap \AD_{\geq t\frac{n}{r}}(R)$.

\begin{lemma}\label{lem-order-intersection}
If $3508\leq r\leq \frac{n}{27\log n}$, $0\leq\ell\leq 27\log r+10$ and $0\leq t\leq 27\log r$, then
\begin{ceqn}
	\begin{align*}
E[|\AD_{\geq t\frac{n}{r}}^\ell(R)|]=O(\frac{r}{\ell!2^t}).
\end{align*}
\end{ceqn}
This also implies that $E[|\AD^\ell(R)|]=O(\frac{r}{\ell!})$ and $E[|\AD_{\geq t\frac{n}{r}}(R)|]=O(\frac{r}{2^t})$.
\end{lemma}
\begin{proof}
	Let $\VD^{\leq 10}(R)$ be the set of prisms in $\VD(R)$ that lie fully above at most $10$ surfaces in $R$,
	and  $\VD^{\leq 10}_{\geq t\frac{n}{r}}(R)$ be the set of prisms in $\VD^{\leq 10}(R)$ intersect at least $t\frac{n}{r}$ surfaces in $F$.
	We first adopt Lemma~\ref{lem-agarwal1} to prove $E[|\VD^{\leq 10}_{\geq t\frac{n}{r}}(R)|]=O(\frac{r}{2^t})$ and then use Theorem~\ref{thm-many-local-less-global} to show $E[|\AD_{\geq t\frac{n}{r}}^\ell(R)|]=O(\frac{r}{\ell!2^t})$.
	
	By \cite[Lemma~5.1]{KaplanMRSS17}, we have $|\VD^{\leq 10}(S')|=O(|S'|)$ for any subset $S'\subseteq S$. If $t=0$, $|\VD^{\leq 10}_{\geq t\frac{n}{r}}(R)|=|\VD^{\leq 10}(R)|=O(r)=O(\frac{r}{2^t})$. 
	Define a conflict between a surface and a prism as the surface intersects the prism, and let $C^0(R)$ be $\VD^{\leq 10}(R)$, so that $C^0_{\geq t\frac{n}{r}}(R)=\VD^{\leq 10}_{\geq t\frac{n}{r}}(R)$.
	For $t\geq 1$, Lemma~\ref{lem-agarwal1} implies that 
	$E[|C^0_{\geq t\frac{n}{r}}(R)|]=O(2^{-t}\cdot E[|C^0(R')|])=O(2^{-t}\cdot |\VD^{\leq 10}(R')|)=O(2^{-t}\cdot \frac{r}{t})=O(\frac{r}{2^t})$,
	where $R'$ is an $\lceil\frac{r}{t}\rceil$-element random subset of $R$.
	
	If $\ell\leq 10$, $E[|\AD^\ell_{\geq t\frac{n}{r}}(R)|]\leq E[|\AD^{\leq 10}_{\geq t\frac{n}{r}}(R)|]\leq E[|\VD^{\leq 10}_{\geq t\frac{n}{r}}(R)|]=O(\frac{r}{2^t})=O\big(\frac{r}{\ell!2^t}\big)$. 
	Redefine a conflict between a surface and a prism as the surface lies fully below the prism, and let $C(R)$ be the set of prisms in $\VD(R)$ that intersect at least $t\frac{n}{r}$ surfaces in $F$, i.e. $C(R)=\VD_{\geq t\frac{n}{r}}(R)$, so that $C^\ell(R)=\VD^\ell_{\geq t\frac{n}{r}}(R)$, $C_{\leq \frac{n}{r}}(R)=\AD_{\geq t\frac{n}{r}}(R)$ and $C_{\leq \frac{n}{r}}^\ell(R)=\AD^\ell_{\geq t\frac{n}{r}}(R)$.\jr{are all three necessary?}
	For $\ell>10$, Theorem~\ref{thm-many-local-less-global} implies
	$E[|C_{\leq \frac{n}{r}}^\ell(R)|]= O\big(\frac{1}{\ell!}\cdot \sum_{i=0}^{10} E[|C^i(R_i)|]\big)$,
	where $R_i$ is an $(r-\ell+i)$-element random subset of $R$.
	Since $E[|C^i(R_i)|]=E[|\VD^i_{\geq t\frac{n}{r}}(R_i)|]\leq E[|\VD^{\leq 10}_{\geq t\frac{n}{r}}(R_i)|]=O(\frac{r-\ell+i}{2^t})$,
	\[E[|\AD_{\geq t\frac{n}{r}}^\ell(R)|]=E[|C_{\leq \frac{n}{r}}^\ell(R)|]=O\big(\frac{1}{\ell!}\cdot \sum_{i=0}^{10} E[|C^i(R_i)|]\big)=O\big(\frac{1}{\ell!}\cdot \sum_{i=0}^{10} \frac{r-\ell+i}{2^t}\big)=O(\frac{r}{\ell!2^t}).\]
	
	The lower bound for the value of $r$ is due to the requirement of Lemma~\ref{lem-agarwal1}. Since $t$ must be at most $\frac{|R_i|}{d}$, $r$ needs to satisfies $27\log r \leq \frac{r-(27\log r+10)}{10}$. The upper bound for the value of $r$ will be applied in the proof of Theorem~\ref{thm-cutting-size} to show that
	with probability at least $1-\frac{1}{r^{17}}$, $\ell\leq 27\log r+10$ and $t\leq 27\log r$.
	Therefore, the exception contributes only $O(1)$ to both $E[|\AD^\ell(R)|]$ and $E[|\AD_{\geq t\frac{n}{r}}(R)|]$, and we can conclude that \[E[|\AD^\ell(R)|]=\sum_{t \geq 0}O(\frac{r}{\ell!2^t})=O(\frac{r}{\ell!})\mbox{ and }E[|\AD_{\geq t\frac{n}{r}}(R)|]=\sum_{\ell\geq 0}O(\frac{r}{\ell!2^t})=O(\frac{r}{2^t}).\]	
\end{proof}

Finally, we adopt Lemma~\ref{lem-order-intersection} to bound $E[|\SC(R)|]$ as follows.

\begin{theorem}\label{thm-cutting-size}
If $3508\leq r\leq \frac{n}{27\log n}$,  then
$E[|\SC(R)|]=O(r)$.
\end{theorem}
\begin{proof}
	We mainly claim that for a prism in $\AD(R)$ lying fully above $\ell$ surfaces in $R$ and intersecting at least $t\frac{n}{r}$, but at most $(t+1)\frac{n}{r}$ surfaces in $F$,
	it contributes $O\big((t+1)^{10}(1+\ell^4)\big)$ to the value of $E[|\SC(R)|]$.
	Since the expected number of such prisms is $O\big(\frac{r}{\ell!2^t}\big)$ (Lemma~\ref{lem-order-intersection}),
	we can conclude the statement as follows:
	\[
	E[|\SC(R)|]=\sum_{\ell\geq 0} \sum_{t\geq 0}O((t+1)^{10}(1+\ell^4)\cdot \frac{r}{\ell!2^t})=O\bigg(r\cdot\sum_{\ell\geq 1}\frac{\ell^4}{\ell!}\underbrace{\sum_{t\geq 1}\frac{t^{10}}{2^t}}_{O(1)}\bigg)=O(r).\]
	The remaining technical details are to prove the claim.
	
	By the construction, $|\SC(R)|$ is linear in the size of the upper envelope formed by the top faces of all prisms in $\RD(R)$. 
	We say two prisms in $\RD(R)$ \emph{vertically overlap} if their projections onto the $xy$-plane intersect.
	A vertical overlapping between two prisms contributes $O(1)$ to the size of the upper envelope since the $xy$-projections of the top faces of any two prisms intersect $O(s)$ times at their boundaries and $s$ is a constant (defined in Section~\ref{sub-sc-distance-VD}).
	Thus, the size of the upper envelope is linear in $|\RD(R)|$ plus the total number of vertical overlappings between prisms in $\RD(R)$.

	It is sufficient to consider the case that all prisms in $\VD(R)$ intersect at most $27\frac{n}{r}\log r$ surfaces in $F$ and all prisms in $\AD(R)$ lies fully above at most $27\log r+10$ surfaces in $R$.
	Lemma~\ref{lem-high-probability} (with $C^0(R)=\VD(R)$, a conflict between a surface and a prism being defined as the surface intersects the prism, $d=10$ and $c=27$) implies that the former condition fails with probability $O(r^{-17})$.
	Moreover, Theorem~\ref{thm-log-level} (with $C(R)=\VD(R)$, a conflict between a surface and a prism as the surface lies fully below the prisms, i.e., $C_{\leq \frac{n}{r}}(R)=\AD(R)$, $c=1$, $d=10$ and $m=17$) implies that the latter condition also fails with probability $O(r^{-17})$. 
	Since $\AD(R)$ contains $O(r^4)$ prisms (when $\AD(R)=\VD(R)$)
	and each prism in $\AD(R)$ can be refined into $O(r^4\log^4 r)$ prisms for $\RD(R)$ (when intersecting with $n$ surfaces), 
	the number of vertical overlappings between prisms in $\RD(R)$ in the worst case is $O\big((r^4\cdot r^4\log^4 r)^2\big)=O(r^{17})$, so that the exception contributes only $O(r^{17})\cdot O(r^{-17})=O(1)$ to the value of $E[|\SC(R)|]$.

	To count the vertical overlapping between two prisms, we charge the prism lying fully above the other. 
	A prism $\Diamond_1$ is said to \emph{lie vertically below} a prism $\Diamond_2$ if $\Diamond_1$ and $\Diamond_2$ vertically overlap and $\Diamond_1$ lies below $\Diamond_2$. 
	For each prism $\Diamond^*\in \RD(R)$, 
	the prisms in $\RD(R)$ lying vertically below $\Diamond^*$ can be categorized into two cases: either (1) generated from the same prism in $\AD(R)$ as $\Diamond^*$ is (i.e., they and $\Diamond^*$ are refined from the same prism in $\AD(R)$) 
	or (2) generated from a prism in $\AD(R)$ lying vertically below $\Diamond^*$.

	Consider a prism $\Diamond \in \AD(R)$ that lies fully above $\bm{\ell}$ surfaces in $R$ and intersects at least $t\frac{n}{r}$, but at most $(t+1)\frac{n}{r}$ surfaces in $F$. 
	By the refinement step in Section~\ref{sub-sc-design},
	$\Diamond$ will be refined into $O\big((t+1)^4\log^4(t+1)\big)=O\big((t+1)^5\big)$ prisms, so that $\Diamond$ contributes $O\big(((t+1)^5)^2\big)=O\big((t+1)^{10}\big)$ vertical overlappings to the quantity of the first case.
	Let $Q$ be the set of prisms in $\AD(R)$ lying vertically below $\Diamond$ and let $Q^*$ be the set of prisms in $\RD(R)$ generated from the prisms in $Q$.
	Then, $\Diamond$ contributes $O\big((t+1)^5\cdot |Q^*|\big)$ vertical overlappings to the quantity of the second case.
	In other words, $\Diamond$ contributes $O\big((t+1)^{10}+(t+1)^5\cdot |Q^*|\big)$ to the value of $E[|\SC(R)|]$.
	
	We will show that $E[|Q^*|]=O\big((t+1)^5\cdot \ell^4\big)$.
	Let $D(\Diamond)$ be the set of surfaces in $F$ that define $\Diamond$, let $K(\Diamond)$ be the set of surfaces in $F$ that intersect $\Diamond$,
	let $L(\Diamond)$ be the set of surfaces in $F$ that lie fully below $\Diamond$,
	and let $F'$ be $F\setminus \big(D(\Diamond)\cup K(\Diamond)\cup L(\Diamond)\big)$.
	We have $|D(\Diamond)|\leq d=10$, $|L(\Diamond)|\leq \frac{n}{r}$ (since $\Diamond\in \AD(R)$), and $t\frac{n}{r}\leq |K(\Diamond)| \leq (t+1)\frac{n}{r}$. Note that $D(\Diamond)\cap L(\Diamond)$ could be nonempty, and recall that it is sufficient to consider $\ell\leq 27\log r+10$ and $|K(\Diamond)| \leq 27\frac{n}{r}\log r$.

	Let $Q_h$ be the set of prisms in $Q$ intersecting with \emph{at least} $h\cdot\frac{|F'|}{r-|D(\Diamond)|-(\ell-|D(\Diamond)\cap L(\Diamond)|)}$ surfaces in $F'$. For each prism in $Q_h\setminus Q_{h+1}$,  since it may intersect surfaces in $D(\Diamond)\cup K(\Diamond)\cup L(\Diamond)$ and since it intersects at most $(h+1)\frac{|F'|}{r-|D(\Diamond)|-(\ell-|D(\Diamond)\cap L(\Diamond)|)} =O\big((h+1)\frac{n}{r}\big)$ surfaces in $F'$, it intersects at most $10+(t+1)\frac{n}{r}+\frac{n}{r}+O\big((h+1)\frac{n}{r}\big)=O\big((h+t+1)\frac{n}{r}\big)$ surfaces in $F$, so that it will be refined into $O((h+t+1)^5)$ prisms.\jr{rephrasing the above three sentences}
	Therefore, we have 
	\[|Q^*|=O\big(\sum_{h\geq 0} |Q_h\setminus Q_{h+1}| \cdot O((h+t+1)^5)\big)=O\big(\sum_{h\geq 0}|Q_h|\cdot (h^5+t^5+1)\big).\]
	
	We will apply Lemma~\ref{lem-agarwal1} to show that $E[|Q_h|]=O(2^{-h}\cdot E[|Q|])$. 
	Due to the assumption of $\Diamond$, 
	$R$ must contain all surfaces in $D(\Diamond)$, $\ell-|D(\Diamond)\cap L(\Diamond)|$ surfaces in $L(\Diamond)\setminus D(\Diamond)$, and no surface of $K(\Diamond)$, so that the corresponding random experiment is to pick $r-|D(\Diamond)|-(\ell-|D(\Diamond)\cap L(\Diamond)|)$ surfaces from $F'$.
	In this situation, Lemma~\ref{lem-agarwal1} (in which $S$ is $F'$, and for any subset $S'\subseteq S$, $C^0(S')$ is the set of prisms in $\VD\big(S'\cup \big(D(\Diamond)\cup L')\big)\big)$ lying vertically fully below $\Diamond$, where $L'$ is an $(\ell-|D(\Diamond)\cap L(\Diamond)|)$-element random subset of $L(\Diamond)\setminus D(\Diamond)$) implies that $E[|Q_h|]=O(2^{-h}\cdot E[|Q|])$,
	leading to that
	\begin{equation*} 
	\begin{split}
	E[|Q^*|] & =O\bigg(\sum_{h\geq 0}\big((2^{-h}\cdot E[|Q|])\cdot (h^5+t^5+1)\big)\bigg)=O\big(((t+1)^5)\cdot E[|Q|]\cdot\sum_{h\geq 0}\frac{(h+1)^5}{2^h}\big)\\
	         & = O((t+1)^5\cdot E[|Q|]).
	\end{split}
	\end{equation*}
	
	We bound $|Q|$ with $O(\ell^4)$. Since $\ell$ surfaces in $R$ lies fully below $\Diamond$, let $\tilde{R}_{\Diamond}$ denote the set of those $\ell$ surfaces, i.e., $\tilde{R}_{\Diamond}=L(\Diamond)\cap R$. 
	Note that all the prisms in $Q$ belong to $\AD(R)$ and lie vertically below $\Diamond$. 
	For each prism $\Diamond_1\in Q$, if we only consider its part lying vertically below $\Diamond$, this part must coincide with the part of a prism $\Diamond_2\in \VD(\tilde{R}_{\Diamond^*})$ lying vertically below $\Diamond$. In other words, if we extend the top face of $\Diamond$
	to be a downward semi-unbounded prism, i.e., without a bottom face,
	its intersection with $\Diamond_1$ is exactly its intersection with $\Diamond_2$.
	Therefore, $|Q|\leq |\VD(\tilde{R}_{\Diamond})|=O(\ell^4)$,
	so that $E[|Q^*|]=O\big((t+1)^5\cdot \ell^4\big)$. To conclude, $\Diamond$ contributes $O\big((t+1)^{10}+(t+1)^5((t+1)^5\cdot \ell^4)\big)=O\big((t+1)^{10}(1+\ell^4)\big)$ to $E[|\SC(R)|]$.
\end{proof}

\section{Data Structure}\label{sec-DS}

Given a set $F$ of $n$ surfaces as in Section~\ref{sub-sc-distance-VD},
generate a sequence of random subsets of $F$, $R_1\subset R_2\subset R_3\subset \ldots \subset R_m$,
where  $|R_i|=2^{i+11}$ for $1\leq i \leq m$ and $\frac{n}{64\log n}<|R_m|\leq \frac{n}{32\log n}$, and let $r_i$ be $|R_i|$.
The $O(\log n)$ shallow cuttings, $\SC(R_1),\SC(R_2), \ldots, \SC(R_m)$, directly yield a data structure for the $k$ nearest neighbors problem with
$O(n\log n)$ space, $O(\log n+k)$ query time, and expected $O\big(n\log^3n\lambda_{s+2}(\log n)\big)$ preprocessing time.
First, since $E[|\SC(R_i)|]=O(r_i)$ (Theorem~\ref{thm-cutting-size}) and each prism in $\SC(R_i)$ stores $O(\frac{n}{r_i})$ surfaces, the expected space is $O\big(\sum_{i=1}^m( r_i \cdot\frac{n}{r_i})\big)=O(n\log n)$. 
By Markov's inequality, with probability at least half, the space is at most twice its expected value, which is also $O(n\log n)$.
Therefore, we can repeat the whole construction until the space is at most twice its expected value,
and the expected number of repetitions is 2, i.e., $O(1)$ repetitions makes the space deterministic.

Second,  if $\frac{n}{r_{i+1}}<k\leq \frac{n}{r_i}$, the query locates the prism $\Diamond\in \SC(R_i)$ intersected by the query vertical line (i.e., the vertical line passing through the query point),
and selects the $k$ lowest surfaces from the $O(\frac{n}{r_i})$ surfaces stored in $\Diamond$. If $k< 32\log n$, search $\SC(R_m)$,
and if $k> \frac{n}{256}$, check $F$ directly. 
Since the \mbox{$xy$-projections} of prisms in $\SC(R_i)$ do not overlap, a planar point-location data structure 
can locate $\Diamond$ in $O(\log n)$ time, and since the selection trivially takes $O(\frac{n}{r_i})=O(2\frac{n}{r_{i+1}})=O(k)$ time, the query time is $O(\log n+k)$.
Finally, since $E[|\SC(R_i)|]=O(r_i)$,
the $m$ point-location data structures can be constructed in expected $\sum_{i=1}^{m}O(r_i)=O(2\cdot r_m)=O(\frac{n}{\log n})$ time~\cite{EdelsbrunnerGS86},
and by Section~\ref{sec-algorithm}, $\SC(R_1),\SC(R_2), \ldots, \SC(R_m)$ can be computed in expected $O\big(n\log^3n\lambda_{s+2}(\log n)\big)$ time.

By Afshani and Chan's two ideas (\cite[Proposition~2.1, Theorem~3.1]{AfshaniC09}),
the space can be further improved using an $O(n)$-space data structure for the \emph{circular range query} problem 
as a secondary data structure. Assume that the preprocessing time and the query time of such a data structure are
$O(n\log n)$ and $O(\kappa+g(n))$, respectively, where $\kappa$ is the number of sites inside the query circular range and $g(\cdot)$ is a function with $g(O(n))=O(g(n))$ and $g(n)\leq n/2$. Let $m'$ be the smallest integer such that $g^{(m'-1)}(n)< 32\log n$. 
The first idea is to store, only for $m'$ shallow cuttings, each prism together with the surfaces intersecting it, 
and to store, for the other $m-m'$ shallow cuttings, only the prisms. 
Consider a subsequence $(R'_1, \ldots, R'_{m'})$ of $(R_1, \ldots, R_m)$
where $R'_{1}$ is $R_1$, $R'_{m'}$ is $R_m$, and $|R'_{i'}|\sim \frac{n}{g^{(i'-1)}(n)}$ for $2\leq i'\leq m'-1$.
For $1\leq i'\leq m'$ and for each prism $\Diamond'\in\SC(R'_{i'})$, we build the circular range query data structure for
the surfaces stored in $\Diamond'$, namely for the respective sites under the given distance measure. 
The expected space is $O(\sum_{i=1}^{m} r_i+\sum_{i'=1}^{m'}\frac{n}{|R'_{i'}|}\cdot |R'_{i'}|)$ $=O(m'n)$;
by the reasoning in the end of the first paragraph,
expected $O(1)$ repetitions would yield deterministic $O(m'n)$ space.

The second idea is to conduct the query in two steps. In the first step, if $\frac{n}{r_{i+1}}<k\leq \frac{n}{r_i}$, we locate the prism in $\SC(R_i)$ intersected by the query vertical line, and find the intersection point between the vertical line and the top face of the prism. The $xy$-projection and the $z$-coordinate of this intersection point decide, respectively, the center and the radius of a circular range. It is clear that this circular range contains $O(\frac{n}{r_i})=O(2\frac{n}{r_{i+1}})=O(k)$ surfaces.
In the second step, if  $\frac{n}{|R'_{i'+1}|}<k\leq \frac{n}{|R'_{i'}|}$, we locate the prism $\Diamond'\in \SC(R'_{i'})$ intersected by the query vertical line,
conduct the above-defined circular range query on the surfaces stored in $\Diamond'$, i.e., on the respective sites, and find the $k$ lowest surfaces from the $O(k)$ surfaces inside the circular range. 
The query time is $O(\log n+g(\frac{n}{|R'_{i'}|})+k)=O(\log n+g(g^{(i'-1)}(n))+k)=O(\log n +\frac{n}{|R'_{i'+1}|}+k)=O(\log n +k)$.

Finally, we show how to make $m'$ be $O(\log\log n)$.
Since the geometric sites in $S$ are of constant description complexity,
they can be lifted to points in $\R^d$ for a sufficiently large constant $d$, e.g., a line segment is mapped to a point in $\R^4$,
and since the distance measure is also of constant description complexity,
the lifted image of each circular range can be described by a constant number of  $d$-variate functions of constant maximum degree.
Therefore, Agarwal et~al.'s algorithm~\cite{AgarwalMS13} yields a circular range query data structure with $O(n)$ space, $O(n^{1/d}+\kappa)$ query time, and $O(n\log n)$ preprocessing time.
Since $g(n)=n^{1/d}$ and $m'$ is the smallest integer such that $g^{(m'-1)}(n)< 32\log n$,
we have $m'=O(\log\log n)$, concluding the following theorem.

\begin{theorem}\label{thm-DS-static}
	Given a distance measure and $n$ geometric sites in the plane as Section~\ref{sub-sc-distance-VD},
	there exists a static data structure for the $k$ nearest neighbors problem with $O(n\log\log n)$ space, $O(\log n+k)$ query time, and expected $O(n\log^3n\lambda_{s+2}(\log n))$ preprocessing time. 
\end{theorem}

By replacing the shallow cuttings in Kaplan et~al.'s dynamic data structure~\cite{KaplanMRSS17} with our shallow cuttings,
we obtain a dynamic data structure as the following corollary. 

\begin{corollary}\label{cor-DS-dynamic}
	There exists a dynamic data structure for the $k$ nearest neighbors problem with $O(n\log n)$ space, $O(\log^2n+k)$ query time, and expected amortized $O(\log^5n\lambda_{s+2}(\log n))$ insertion time, and expected amortized $O(\log^7n\lambda_{s+2}(\log n))$ deletion time.
\end{corollary}

\section{Construction Algorithm}\label{sec-algorithm}

Given a set $F$ of $n$ surfaces as in Section~\ref{sub-sc-distance-VD},
consider a sequence of random subsets of $F$, $R_1\subset R_2\subset R_3\subset \ldots \subset R_m$,
where $|R_i|=2^{i+11}$ for $1\leq i \leq m$
and $\frac{n}{64\log n}<|R_m|\leq \frac{n}{32\log n}$, let $r_i$ be $|R_i|$.
It is clear that $|R_1|=4096$ and $m\leq \log n-5\log\log n$. 

We will build the $\frac{1}{r_i}$-shallow-cuttings $\SC(R_i)$ for $1\leq i\leq m$
in the following three steps:
\begin{enumerate}
\item Repeatedly generate $R_1\subset R_2\subset R_3\subset \ldots \subset R_m$
and construct $\VD^{\leq \ell}(R_i)$ for $1\leq i\leq m$, where $\ell=\Theta(\log n)$ and $\VD^{\leq \ell}(R_i)$ is the set of prisms in $\VD(R_i)$ whose level with respect to $R_i$ is at most $\ell$,\footnote{Kaplan et~al.~\cite{KaplanMRSS17} used the term $\VD_{\leq \ell}(R)$, while in this paper, to be consistent with Section~\ref{sec-random-sampling}, we always use subscripts to describe global relations and superscripts to describe local relations.} until $\VD^{\leq \ell}(R_i)$ includes $\AD(R_i)$  for $1\leq i\leq m$.  (Recall that $\AD(R_i)$ is the set of prisms in $\VD(R_i)$ lying fully above at most $\frac{n}{r_i}$ surfaces in $F$.)
\item Generate $\AD(R_i)$ from $\VD^{\leq \ell}(R_i)$ for $1\leq i\leq m$. 
\item Build a $\frac{1}{r_i}$-shallow-cutting $\SC(R_i)$ from $\AD(R_i)$ in the same way as Section~\ref{sub-sc-design} for $1\leq i\leq m$.
\end{enumerate}
For simplicity, we use $R$ to denote $R_i$ when the context is clear,
and let $r$ be $|R|$.

We will make use of an auxiliary data structure $\bm{\RDS}$ by Kaplan et~al.\ (\cite[Theorem 7.1]{KaplanMRSS17}) 
with expected $O(n\log^2 n)$ preprocessing time and expected $O(\log n+k)$ query time for the following two types of queries: (a) given a query point $p\in\R^3$,
answer the $k$ surfaces in $F$ lying below $p$, where $k$ is unknown,
and (b) given a query vertical line and an integer $k$, answer the lowest $k$ surfaces in $F$ along the line.
Although Kaplan et~al.\ only mentioned the first type,
since their data structure is generalized from Chan's data structure for planes~\cite{Chan00},
it directly works for the second type. 

\subsection{Construct $\VD^{\leq \ell}(R)$}\label{sub-al-construct-VD}

Let $(f_1,f_2,\ldots, f_n)$ be a random sequence of $F$, and let $F_j$ be $\{f_1, f_2,\ldots, f_j\}$,
so that $R_i=F_{2^{i+11}}$ for $1\leq i\leq m$. 
Kaplan et~al.\ (\cite[Section~5]{KaplanMRSS17}) proved the size of $\VD^{\leq \ell}(F_j)$ to be $O(j\cdot \ell\cdot \lambda_{s+2}(\ell))$
and constructed, in expected $O(n\log^2n \ell \lambda_{s+2}(\ell))$ time, $\VD^{\leq \ell}(F_j)$ for $1\leq j\leq n$
together with for each prism $\Diamond\in \VD^{\leq \ell}(F_j)$ the set $F_\Diamond$ of surfaces in $F$ intersecting $\Diamond$.
Let $\ell$ be $30\log n$, so that the running time is expected $O(n\log^3n\lambda_{s+2}(\log n))$, and with high probability, $\AD(R_i)\subseteq \VD^{\leq \ell}(R_i)$ for $1\leq i\leq m$, the latter of which will be explained in the time analysis of Section~\ref{sub-al-time}. 

To verify if $\AD(R)\subseteq \VD^{\leq\ell}(R)$,
we actually examine the ``upper envelope'' of  $\VD^{\leq\ell}(R)$.
Precisely, we consider each prism $\Diamond$ in $\VD^\ell(R)$, where $\VD^\ell(R)$ is the set of prisms in $\VD(R)$ whose level with respect to $R$ is exactly $\ell$,
and check if the top face $\triangle$ of $\Diamond$ lies fully above at least $\frac{n}{r}$ surfaces in $F$.
For the latter,
we first derive the set $F_\triangle$ of surfaces in $F$ that intersects $\triangle$ from $F_\Diamond$, where $F_\Diamond$ is the set of surfaces in $F$ that intersects $\Diamond$.
Then, we arbitrarily pick a point $p\in \triangle$, and use $\RDS$ to find the $|F_\triangle|+\frac{n}{r}$ lowest surfaces along the vertical line passing through $p$. 
Since a surface lying fully below $\triangle$ lies below $p$ and since a surface lying below $p$ either intersects $\triangle$ or lies fully below $\triangle$,
$\triangle$ lies fully above at least $\frac{n}{r}$ surfaces 
if and only if 
$\triangle$ lies fully above at least $\frac{n}{r}$ surfaces of the returned $|F_\triangle|+\frac{n}{r}$ ones. 
Therefore, if $\triangle$ lies fully above at least $\frac{n}{r}$ returned surfaces, we say $\triangle$ pass the test,
and if all the top faces of prisms in $\VD^\ell(R)$ pass the test, we determine that  $\AD(R)\subseteq \VD^{\leq\ell}(R)$.

If we do not determine that $\AD(R_i) \subseteq \VD^{\leq\ell}(R_i)$ for some $i\in [1, m]$,
we generate a new random sequence $(f_1, f_2,\ldots,f_n)$ of $F$, and repeat the above process.

\subsection{Select $\AD(R)$ from $\VD^{\leq\ell}(R)$}\label{sub-al-construct-AD}

To select $\AD(R)$ from $\VD^{\leq\ell}(R)$,
we conduct a test on each prism $\Diamond\in\VD^{\leq\ell}(R)$, which is similar to the one in Section~\ref{sub-al-construct-VD}.
We arbitrarily pick point $p\in \Diamond$, use $\RDS$ to find the $|F_\Diamond|+\frac{n}{r}+1$ lowest surfaces along the vertical line passing through $p$,
and if $\Diamond$ lies fully above at most $\frac{n}{r}$ returned surfaces,
include $\Diamond$ in $\AD(R)$.
Again,
since a surface lying fully below $\Diamond$ lies below $p$ and since a surface lying below $p$ either intersects $\Diamond$ or lies fully below $\Diamond$,
$\Diamond$ lies fully above at most $\frac{n}{r}$ surfaces 
if and only if 
$\Diamond$ lies fully above at most $\frac{n}{r}$ surfaces of the returned $|F_\Diamond|+\frac{n}{r}+1$ ones.

\subsection{Build $\SC(R)$ from $\AD(R)$}\label{sub-al-construct-SC}

The procedure is already outlined in Section~\ref{sub-sc-design},
so  we provide the implementation details.

\subsubsection{Build $\RD(R)$ from $\AD(R)$}\label{subsub-al-refinement}

For each prism $\Diamond\in\AD(R)$,
let $F_\Diamond$ be the set of surfaces in $F$ intersecting $\Diamond$,
and let $t$ be $\lceil \frac{|F_\Diamond|r}{n}\rceil$.
If  $t>1$,
we refine $\Diamond$ into smaller prisms by
picking an $O(t\log t)$-element random subset $F'$ of $F_\Diamond$,
and applying Chazelle et~al.'s algorithm~\cite{ChazelleEGS91} to construct $\VD(F')\cap \Diamond$,
which takes $O(|F'|^5\log |F'|)=O\big((t^5\log^5t)\cdot \log (t\log t)\big)=O(t^6)$ time and generates $O(|F'|^4)=O(t^4\log^4t)=O(t^5)$ prisms.
For each prism $\Diamond'\in \VD(F')\cap \Diamond$,
we generate $F_{\Diamond'}$ by testing surfaces in $F_\Diamond$ with $\Diamond'$, which takes $O(|F_\Diamond|\cdot t^5)=O(t^6\cdot \frac{n}{r})$ time. 
By defining a conflict between a surface and a prism as the surface intersects the prism, 
Corollary~\ref{cor-high-probability} implies that
with probability at least half, each prism $\Diamond'\in\VD(F')\cap\Diamond$ intersects at most $\frac{|F_\Diamond|}{t}\leq \frac{n}{r}$ surfaces in $F$, i.e., $|F_{\Diamond'}|\leq \frac{n}{r}$,
so that we can repeat the process until the requirement is satisfied and the expected number of repetitions is $O(1)$.
Finally, for each prism $\Diamond'\in \VD(F')\cap \Diamond$,
we conduct the same test as in Section~\ref{sub-al-construct-AD} to check if $\Diamond'$ lies fully above at most $\frac{n}{r}$ surfaces in $F$,
and if no, discard $\Diamond'$,
which takes expected $O(\log n+|F_{\Diamond'}|+\frac{n}{r}+1)=O(\log n+\frac{n}{r})$ time, i.e., $O(t^5\cdot (\log n+\frac{n}{r}))$ in total. 
Therefore, the refinement of $\Diamond$ takes expected $O(t^6\cdot \frac{n}{r}+t^5\log n)$ time.

\subsubsection{Build $\SC(R)$ from $\RD(R)$}\label{subsub-al-vertical}

First, we construct the upper envelope of top faces of prisms in $\RD(R)$ through an algorithm in \cite[Section~7.3.4]{SA95}, which is a mixture of divide-and-conquer and plane-sweep methods. 
Then, we partition each face of the upper envelope into trapezoids using Chazelle's linear time algorithm~\cite{Chazelle91},
and extend these resulting trapezoids to downward semi-unbounded prisms, leading to $\SC(R)$.
Finally, for each prism $\Diamond'\in \SC(R)$, 
we will build $F_{\Diamond'}$, i.e., the set of surfaces in $F$ intersecting $\Diamond'$. 

To build $F_{\Diamond'}$, it is clear that each surface in $F_{\Diamond'}$ either intersects with or lies fully below the top face of $\Diamond'$. 
Let $\Diamond$ be the prism in $\RD(R)$ whose top face contains the top face of $\Diamond'$.
We check each surface in $F_\Diamond$ with $\Diamond'$ to find the surfaces intersecting the top face of $\Diamond'$.
Moreover, we arbitrarily pick a point $p$ in the top face of $\Diamond'$,
and use $\RDS$ to find all surfaces in $F$ lying below $p$, from which we can find the surfaces lying fully below the top face of $\Diamond'$. 
The above two steps take $O(|F_\Diamond|)$ time and expected $O(\log n+|F_{\Diamond'}|)$ time, respectively.
Since $|F_\Diamond|\leq \frac{n}{r}$ and $|F_{\Diamond'}|\leq2\frac{n}{r}$,
the generation of $F_{\Diamond'}$ takes expected $O(\log n+\frac{n}{r})$ time. ($\Diamond$ belongs to $\RD(R)$, so $|F_\Diamond|\leq \frac{n}{r}$; as discussed in the end of Section~\ref{sub-sc-design}, $|F_{\Diamond'}|\leq 2\frac{n}{r}$.)

\subsection{Time Analysis}\label{sub-al-time}

For the construction time,
we analyze the three steps separately.
In addition to Lemma~\ref{lem-agarwal1}, we will also make use of another Agarwal et~al.'s result~\cite{AgarwalMS98} as follow. 

\begin{lemma}\label{lem-agarwal2}(\cite[Proposition~2.1]{AgarwalMS98})
	For an $r$-element random subset $R$ of $S$, if $C^0(R)$ satisfies Conditions~(i)~and~(ii) (in Section~\ref{sub-rs-configuration}) and if every subset $R'\subseteq R$ satisfies $E[|C^0(R')|]=O(f(|R'|))$ for an increasing function $f(\cdot)$, then
	\[ E[\sum_{\triangle\in C^0(R)} w(\triangle)]=O\big ((\frac{n}{r})\cdot f(|R|)\big).\]
\end{lemma}

For the first step, since $\ell$ is $30\log n$, by \cite[Theorem~5.1]{KaplanMRSS17}, 
it takes expected\\
$O(n\log^3n\lambda_{s+2}(\log n))$ time to build $\VD^{\leq \ell}(F_1), \ldots, \VD^{\leq \ell}(F_n)$.
It is clear that the upper envelope of $\VD^{\leq \ell}(R)$
consists of top faces of prisms in $\VD^\ell(R)$.
Since it takes for each prism $\Diamond'\in \VD^\ell(R)$ expected $O(\log n+\frac{n}{r}+|F_{\Diamond'}|)$ time to check if the top face of $\Diamond'$ lies fully above at least $\frac{n}{r}$ surfaces,
the total expected time to check all prisms in $\VD^\ell(R)$ is \[O\big(E[\sum_{\Diamond'\in \VD^\ell(R)}(\log n+\frac{n}{r}+|F_{\Diamond'}|)]\big )\leq O\big( E[\sum_{\Diamond\in \VD^{\leq\ell}(R)}(\log n+\frac{n}{r}+|F_\Diamond|)] \big).\]
Since $|\VD^{\leq\ell}(R)|=O(r\log n\lambda_{s+2}(\log n))$ (\cite[Lemma~5.1]{KaplanMRSS17}),
Lemma~\ref{lem-agarwal2} implies that 
\[E[\sum_{\Diamond\in \VD^{\leq\ell}(R)}|F_\Diamond|]= O(\frac{n}{r}\cdot E[|\VD^{\leq\ell}(R)|])=O\big(\frac{n}{r}\cdot (r\log n\lambda_{s+2}(\log n))\big)=O\big(n\log n\lambda_{s+2}(\log n)\big).\]
As a result, since $r=|R|\leq|R_m|=O(\frac{n}{\log n})$, the expected time to check $\VD^{\ell}(R)$ is
\[O\bigg(\big(\underbrace{E[|\VD^{\leq\ell}(R)|]}_{O(r\log n\lambda_{s+2}(\log n))}\big)(\log n+\frac{n}{r})+n\log n\lambda_{s+2}(\log n)\bigg)=O\big(n\log n\lambda_{s+2}(\log n)\big),\]
implying that the total expected time to check $\VD^{\leq \ell}(R_i)$ for $1\leq i\leq m$ is $O\big(n\log^2 n\lambda_{s+2}(\log n)\big)$ time. 
Thus, the total expected time is $O(n\log^3n\lambda_{s+2}(\log n))$,
dominated by the construction of $\VD^{\leq \ell}(F_j)$ for $1\leq j\leq n$.

Since $\ell=30\log n \geq 27\log n +10$, 
according to the proof of Theorem~\ref{thm-cutting-size},
with probability $1-O(\frac{1}{n^{17}})$, $\AD(R)\subseteq \VD^{\leq \ell}(R)$, 
implying that with probability $1-O(\frac{\log n}{n^{17}})$,
$\AD(R_i)\subseteq \VD^{\leq \ell}(R_i)$ for $1\leq i\leq m$. 
As a result, 
the expected number of repetitions for the first step is $O(1)$,
and the first step takes expected $O\big(n\log^3n\lambda_{s+2}(\log n)\big)$ time in total.

For the second step,
since it takes expected $O(\log n+\frac{n}{r}+|F_{\Diamond}|)$ to time check if a prism $\Diamond\in\VD^{\leq \ell}(R)$ belongs to $\AD(R)$, 
the same analysis for the test in the first step yields that the second step takes expected $O\big(n\log^2n\lambda_{s+2}(\log n)\big)$ time to build $\AD(R_i)$ from $\VD^{\leq\ell}(R_i)$ for $1\leq i\leq m$.

For the third step,
we analyze the refinement and the upper envelope construction separately.
For the refinement, 
as discussed in Section~\ref{subsub-al-refinement},
if a prism $\Diamond\in \AD(R)$ intersects at most $t\frac{n}{r}$ functions,
it takes $O(t^6\cdot \frac{n}{r}+t^5\log n)$ time to process $\Diamond$, and $O(t^5)$ prisms are generated.  
By Theorem~\ref{thm-size-AD-RD},
the expected number of prisms in $\AD(R)$ that intersect at least $(t-1)\frac{n}{r}$ surfaces in $F$
is $O(2^{-(t-1)}\cdot r )$, 
implying that the expected time\jr{to build $\RD(R)$ from $\AD(R)$} is
\[\sum_{t\geq  1}O\big((t^6\frac{n}{r}+t^5\log n)\cdot (2^{-(t-1)}\cdot r)\big)=O(n+r\log n)\cdot\sum_{t\geq 1}t^6\cdot 2^{-(t-1)}=O(n),\]
where the last inequality comes from the fact that $r=O(\frac{n}{\log n})$ and $\sum_{t\geq 1}t^6\cdot 2^{-(t-1)}=O(1)$.

For constructing the upper envelope of $\RD(R)$,
the algorithm in \cite[Section~7.3.4]{SA95} recursively divides $\RD(R)$ into two subsets of roughly equal size,
and use plane-sweep to merge the upper envelopes of the two subsets.
Since $E[|\RD(R)|]=O(r)$ (Theorem~\ref{thm-size-AD-RD}), there are expected $O(\log r)$ recursion levels.
As shown in the proof of Theorem~\ref{thm-cutting-size},
the expected total number of intersections among the boundaries of the $xy$-projections of prisms in $\RD(R)$
is $O(r)$, so that the expected total complexity of upper envelopes in a recursion level is $O(r)$.
Therefore, the plane sweep takes expected $O(r\log r)$ time for one recursion level, and the algorithm takes expected $O(r\log^2 r)$ time to build $\SC(R)$.

After the construction of $\SC(R)$, we need to compute $F_\Diamond$ for each prism $\Diamond\in \SC(R)$,
which takes expected $O(\log n+\frac{n}{r})$ time as discussed in Section~\ref{subsub-al-vertical}.
Since $E[|\SC(R)|]=O(r)$ and $r=O(\frac{n}{\log n})$,
it takes expected $O\big(r\cdot (\log n+\frac{n}{r})\big )=O(n)$ time for all prisms in $\SC(R)$.
The total expected time for the third step is $O\big(\sum_{i=1}^m(n+r_i\log^2r_i)\big)=O(n\log n+r_m\log^2r_m)=O(n\log n)$.

\begin{theorem}\label{thm-construction-time}
	It takes expected $O\big(n\log^3n\lambda_{s+2}(\log n)\big)$ time to compute $\frac{1}{r_i}$-shallow-cuttings $\SC(R_i)$ of $F$ for $1\leq i\leq m$ such that the total size of those $O(\log n)$ cuttings is $O(\frac{n}{\log n})$
	and the total space to store the surfaces intersecting every prism is $O(n\log n)$. 
\end{theorem}
\begin{proof}
	The running time has been analyzed. By Theorem~\ref{thm-cutting-size},
	the expected total size is 
	\[E[\sum_{i=1}^m|\SC(R_i)|]=\sum_{i=1}^m O(r_i)=\sum_{i=1}^m O(2^{i+11})=O(2^{m+12})=O(\frac{n}{\log n}).\]
	Since each prism in $\SC(R)$ intersects $O(\frac{n}{r})$ surfaces in $F$,
	the expected total space is 
	\[\sum_{i=1}^m\big(E[|\SC(R_i)|]\cdot O(\frac{n}{r_i})\big)=\sum_{i=1}^mO(r_i)\cdot O(\frac{n}{r_i})=\sum_{i=1}^mO(n)=O(n\log n).\]
	By Markov's inequality,
	with probability at most $1/4$, the value is more than four times its expectation.
	Therefore, with probability at least $1-(1/4+1/4)=1/2$, both the total size and the total space are at most four times their expected values, respectively.
	We can repeat the whole construction algorithm until the both values are at most four times their expectations,
	i.e., making the bounds for the total size and the total space deterministic, and the expected number of repetitions is only 2. 
\end{proof}

\section{Concluding Remarks}\label{sec-conclusion}

We have derived a new random sampling technique for configuration space, have applied our new technique to successfully design linear-size shallow cuttings for general distance functions, and have composed these shallow cuttings into nearly optimal static and dynamic data structures for the $k$ nearest neighbor problem.
The remaining challenges are the optimal $O(n)$ space and the optimal $O(n\log n)$ preprocessing time. 
Afshani and Chan's $O(n)$ space for point sites in the Euclidean metric~\cite{AfshaniC09} is an elegant combination of Matou\v{s}ek's \emph{shallow cutting lemma} and \emph{shallow partition theorem}~\cite{Matousek92a}.
Although we have designed linear-size shallow cuttings for general distance functions, 
the generalization of ``$O(\log r)$-crossing'' shallow partitions is still unknown since the original proof significantly depends on certain geometric properties of planes.
For the $O(n\log n)$ preprocessing time, the traversal idea by Chan~\cite{Chan00} and Ramos~\cite{Ramos99a}
seems not to work directly since a pseudo-prism is possibly adjacent to a ``non-constant'' number of pseudo-prisms in the vertical decomposition.

\appendix

\section{Literature for Shallow Cuttings}\label{ap-literature}

To study the \emph{half-space range reporting} problem,
Matou\v{s}ek~\cite{Matousek92a} first used ``tetrahedra'' to define shallow cuttings and proved the existence of  a $\frac{1}{r}$-shallow-cutting of $O(r)$ tetrahedra.
Then, he adopted his shallow cuttings to prove \emph{Shallow Partition Theorem}, and applied the theorem to construct a data structure for the half-space range reporting problem.
Ramos~\cite{Ramos99a} developed an $O(n\log n)$ time randomized algorithm to construct $\frac{1}{r}$-shallow-cuttings of tetrahedra for $r=2, 4, 8, \ldots, \frac{n}{\log n}$. 
Chan~\cite{Chan00} observed that those tetrahedra can be turned into disjoint downward semi-unbounded vertical triangular prisms,
so that the resulting shallow cuttings can be applied to the $k$ lower plane problem.

Both Ramos~\cite{Ramos99a} and Chan~\cite{Chan00} adopted the bootstrapping technique to only select $O(\log\log n)$ $\frac{1}{r}$-shallow-cuttings instead of $O(\log n)$ ones.
Since a $\frac{1}{r}$-shallow-cutting requires $O(r)\cdot O(\frac{n}{r})=O(n)$ space,
the above results yield a static data structure for the $k$ lowest plane problem with $O(n\log\log n)$ space, $O(\log n+k)$ query time, and expected $O(n\log n)$ preprocessing time.
Afshani and Chan~\cite{AfshaniC09} further exploited Matou\v{s}ek's \emph{shallow partition theorem} \cite{Matousek92a} to show that only two $\frac{1}{r}$-shallow-cuttings are sufficient, 
leading to the optimal $O(n)$ space.

Chan~\cite{Chan10} designed a dynamic data structure for the $k$ lower plane problem also based on shallow cuttings with $O(\log^2n+k)$ query time, expected amortized $O(\log^3 n)$ insertion time and expected amortized $O(\log^6 n)$ deletion time. 
Chan and Tsakalidis~\cite{ChanT16} proposed a deterministic construction algorithm for the $\frac{1}{r}$-shallow-cuttings, making the above-mentioned time complexities deterministic.
Later, Kaplan et~al.~\cite{KaplanMRSS17} attained the amortized $O(\log^5n)$ deletion time, and very recently, Chan~\cite{Chan19} further improved the deletion time to amortized $O(\log^4n)$.

\section{Applications of the New Dynamic Data Structure}\label{ap-application}

Kaplan et~al.~\cite{KaplanMRSS17} mentioned 9 applications that can be improved directly using their dynamic data structure for neighbor neighbor queries (i.e., with $k=1$). 
Since our dynamic data structure improves Kaplan et~al.'s by a $\log^2 n$ factor in both space and deletion time,
those applications can be further improved by a $\log^2 n$ factor in space, update time (deletion or both insertion and deletion), or construction time.
Table~\ref{tb-direct} and Table~\ref{tb-disk} show the corresponding improvements. For the completeness, we introduce those applications in the following two subsections.

We remind that Kaplan et~al.'s dynamic data structure for the lower envelope of surfaces 
corresponds to our dynamic data structure for the $k$ nearest neighbors problem,
so that we will not distinguish between a dynamic data structure for the lower envelope of surfaces and a dynamic data structure for nearest neighbor queries.
 In fact, a vertical ray shooting query to the lower envelope of $n$ surfaces
is equivalent to a nearest neighbor query to the $n$ respective sites. 
Chan~\cite{Chan10} showed how to use shallow cuttings to maintain the lower envelope of planes dynamically and how to use his data structure to answer $k$ nearest neighbors queries. Kaplan et~al.\ extended his idea to maintain the lower envelope of surfaces using shallow cuttings of semi-unbounded pseudo-prisms. Our dynamic data structure just replaces the shallow cuttings in Kaplan et~al.'s data structure with our designed ones, whose size is linear and reduces a double logarithmic factor from their size. 

Two classes of distance functions will be applied. First, let $p\in[1,\infty]$; for two points $(x_1,y_1)$, $x_2,y_2)\in \R^2$, their distance in the $L_p$ norm is $(|x_1-x_2|^p+|y_1-y_2|^p)^{1/p}$. Second, let $S$ be a set of point sites in $\R^2$ and associate each site $q\in S$ a weight $w_q\in \R$;
the additively weighted Euclidean distance from a point $p\in \R^2$ to a site $q\in S$ is $w_q+|\overline{pq}|$, where $|\cdot|$ denotes the Euclidean distance.

\subsection{Direct Applications}\label{sub-direct-ap}

\subparagraph{Dynamic Bichromatic Closest Pair.}
Let $\tau$ be a planar distance metric, and let $R$ and $B$ be two sets of point sites in the plane. 
A \emph{bichromatic closest pair} between $R$ and $B$ is a pair of points, $r\in R$ and $b\in B$, that minimizes $\tau(r,b)$. The dynamic version is to maintain a bichromatic closest pair under insertions and deletions of points. Eppstein~\cite{Eppstein95} proved that if there exists a data structure that supports insertions, deletions, and nearest-neighbor queries in $O(T(n))$ time per operation, 
a bichromatic closest pair between $R$ and $B$ can be maintained in $O(T(n)\log n)$ time per insertion and $O(T(n)\log^2n)$ time per deletion. 
Since the insertion time, the deletion time, and the nearest-neighbor query time  in our data structure are $O(\log^{5} n\lambda_{s+2}(\log n))$, $O(\log^{7} n\lambda_{s+2}(\log n))$, and $O(\log n)$ respectively, $T(n)=O(\log^{7} n\lambda_{s+2}(\log n))$,
resulting in $O(\log^{8} n\lambda_{s+2}(\log n))$ time per insertion and $O(\log^{9} n\lambda_{s+2}(\log n))$ time per deletion.

\subparagraph{Minimum Euclidean Bichromatic Matching.}
Let $R$ and $B$ be two sets of $n$ point sites in the plane.  A \emph{minimum Euclidean bichromatic matching} between $R$ and $B$
is a set of $n$ line segments $\overline{rb}$, $r\in R$ and $b\in B$, such that each point site in $R\cup B$ is incident to exactly one line segment and such that the total length of the line segments is minimum. Agarwal et~al.\ (\cite[Section~7]{AgarwalES99}) showed how to find a minimum Euclidean bichromatic matching using a dynamyic bichromatic closest pair data structure for the additively weighted Euclidean metric, and the construction time is $O(n^2\cdot T(n))$, where $T(n)$ is the maximum between the insertion time and the deletion time of the dynamic data structure.
Since $T(n)$ in our new data structure (i.e., the first application)  is $O(\log^{9} n\lambda_{s+2}(\log n))$, the construction time is $O(n^2\log^{9} n\lambda_{s+2}(\log n))$.

\subparagraph{Dynamic Minimum Spanning Trees.}
Let $S$ be a set of sites, and let $T$ be the minimum spanning tree for $S$ with respect to  an $L_p$ norm, $p\geq 1$. A dynamic minimum spanning tree data structure maintains $T$ explicitly as $S$ changes dynamically. Following Eppstein~\cite{Eppstein95}, the first application yields a data structure with $O(n\log^3 n)$ space and $O(\log^{11} n\lambda_{s+2}(\log n))$ update time. 

\subparagraph{Dynamic Intersection of Unit Balls in Three Dimensions.}
Let $B$ be a set of unit balls in $\R^3$. The goal is to maintain the intersection $B^{\cap}$ of the balls in $B$ under insertions and deletions, and in the meantime to support the two following queries:
\begin{enumerate}[label=(\alph*),itemsep=0pt, topsep=6pt]
	\item for any point $p\in \R^3$, determine if $p\in B^{\cap}$,
	\item and after performing each update, determine whether $B^{\cap}=\emptyset$. 
\end{enumerate}
Agarwal et~al.\ (\cite[Section~8]{AgarwalES99}) adopted dynamic lower envelope data structures to maintains $B^{\cap}$. Recall that the dynamic lower envelope data structure correspond to our dynamic nearest neighbor data structure. Since their algorithm performs a query via parametric search in a black-box fashion,
our dynamic data structure for the $k$ nearest neighbors problem implies a dynamic data structure to maintain $B^{\cap}$ with $O(n\log n)$ space, $O(\log^{5} n\lambda_{s+2}(\log n))$ insertion time, $O(\log^{7} n\lambda_{s+2}(\log n))$ deletion time, and $O(\log^2 n)$ and $O(\log^5 n)$ query time respectively for queries (a) and (b).

\subparagraph{Dynamic Smallest Stabbing Disk.}
Let $\mathcal{C}$ be a family of simply-shaped, compact, strictly-convex sets in the plane. The goal is to dynamically maintain a finite subset $C\subseteq \mathcal{C}$ together with a smallest disk that intersects all the sets of $C$; please see Section~9 by Agarwal et~al.~\cite{AgarwalES99} for precise definitions.
Since Agarwal et~al.\ made use of the dynamic data structure for nearest neighbor queries in a black-box fashion (\cite[Theorem~9.3]{AgarwalES99}), 
our dynamic data structure for the $k$ nearest neighbors problem yields a dynamic structure to maintain the smallest stabbing disk with $O(n\log n)$ space, $O(\log^{5} n\lambda_{s+2}(\log n))$ insertion time, $O(\log^{7} n\lambda_{s+2}(\log n))$ deletion time, and $O(\log^5 n)$ query time. 

\subsection{Problems on Disk Intersection Graph}\label{sub-disk-ap}

As discussed by Kaplan et~al.~\cite{KaplanMRSS17}, 
a dynamic data structure for neighbor neighbor queries
enables many applications in the domain of disk intersection graphs: let $S$ be a finite set of point sites in the plane,
and associate each point site $p\in S$ with a weight $w_p>0$. 
The \emph{disk intersection graph} for $S$, denoted by $D(S)$, has $S$ as the vertex set and an edge between two point sites $p,q\in S$ if and only if $|\overline{pq}|\leq w_p + w_q$. In other words, $D(S)$ contains an edge between $p$ and $q$ if and only if the disk with center $p$ and radius $w_p$ intersects the disk with center $q$ and radius $w_q$. If all weights are 1, $D(S)$ is called \emph{unit disk graph} and denoted by $\mathrm{UD}(S)$.
Disk intersection graphs have been received increasing attention due to the applications in wireless sensor networks~\cite{CabelloJ15,ChanPR11,FurerK12,KaplanMRSS17,RodittyS11}.

\subparagraph{Shortest Path Trees in Unit Disk Graphs.}
Cabello and Jejcic~\cite{CabelloJ15} showed how to find a shortest path tree in $\mathrm{UD}(S)$, for any given site in $S$, within $O(n\cdot T(n))$ time using a dynamic
bichromatic closest pair structure, where $T(n)$ is the maximum between the insertion time and the deletion time. Since our first application in Appendix~\ref{sub-direct-ap} attains $T(n)=O(\log^{9} n\lambda_{s+2}(\log n))$,
the construction time is $O(n\log^{9} n\lambda_{s+2}(\log n))$.

\subparagraph{Dynamic Connectivity in Disk Graphs.}
Kaplan et~al.~\cite{KaplanMRSS17}, in Section~9.2 of their full version,
studied how to dynamically maintain $D(S)$ under insertions and deletions of point sites while answering \emph{reachability queries} efficiently: given $s,t\in S$, determine if there is a path in $D(S)$ between $s$ and $t$. Let $\Psi$ be the ratio of the largest and the smallest weights of the point sites. They adopted their dynamic lower envelope structure to attain $O(\Psi^2\log^{9} n\lambda_{s+2}(\log n))$ update time and $O(\log n/\log\log n)$ query time. Since our dynamic lower envelope structure improves the deletion time of their dynamic lower envelope structure by a factor of $\log^2n$, the update time of their dynamic connectivity structure for a disk graph is also improved by a factor of $\log^2n$.

\subparagraph{BFS Trees in Disk Graphs.}
Kaplan et~al.~\cite{KaplanMRSS17}, in Section~9.3 of their full version,
extended Roditty and Segal's observation~\cite{RodittyS11} to  compute exact BFS-trees in disk graphs, for any given root $r\in S$, using a dynamic nearest neighbor query structure. 
In details, they adopted additively weighted Euclidean distances, so that the weighted Euclidean distance from a point site $p$ with weight $w_p$ is exactly the Euclidean distance from the disk with center $p$ and radius $w_p$. Their construction time for a BFS tree  is $O(n\cdot T(n))$, where $T(n)$ is the maximum between the insertion time and the deletion time per point site. Since $T(n)$ is $O(\log^7\lambda_{s+2}(\log n))$ in our dynamic nearest neighbor query structure, the construction time is $O(n\log^7\lambda_{s+2}(\log n))$.

\subparagraph{Spanners for Disk Graphs.}
A $(1+\rho)$-spanner for $D(S)$ is a subgraph $H$ of $D(S)$ such that the shortest path distances in $H$ approximate the shortest path distances in $D(S)$ up to a factor of $(1+\rho)$. Kaplan et~al.~\cite{KaplanMRSS17}, in Section~9.4 of their full version, gave a construction algorithm with $O((n/\rho^2)\log^{9} n\lambda_{s+2}(\log n))$ time using their dynamic nearest neighbor query structure. Since our dynamic data structure improves their deletion time by a factor of $O(\log^2 n)$, the construction becomes $O((n/\rho^2)\log^{7} n\lambda_{s+2}(\log n))$.

\section{Comparison with Clarkson 1987's Result~\cite{Clarkson87}}\label{ap-clarkson}

We first introduce Clarkson 1987's result (\cite[Corollary~4.3]{Clarkson87}) about relatively many local conflicts, but relatively few global conflicts,
then explain the difficulty in applying his result to design a linear-size shallow cutting, and finally compare the state-of-the-art techniques at a high level. 
Recall the definitions in Section~\ref{sub-rs-configuration}:
$C(S')$ is the set of configurations in a geometric structure defined by $S'$, and $T(S')=\bigcup_{S''\subseteq S'}C(S'')$ is the set of all possible configurations defined by objects in $S'$.
More importantly, there is a quantity difference between $C(S')$ and $T(S')$. In the examples in Section~\ref{sub-rs-configuration},
objects are planes,
$C(S')$ is the set of tetrahedra in the canonical triangulation for the arrangement formed by all planes in $S'$,
and $T(S')$ is the set of all possible tetrahedra defined by planes in $S'$, so that  $|C(S')|=\Theta(|S'|^3)$~\cite{AgarwalBMS98,Mulmuley94}, while $|T(S')|=O(|S'|^{12})$ (since a tetrahedron has 4 vertices and a vertex is defined by 3 planes).

Roughly speaking,
our Theorem~\ref{thm-many-local-less-global} works on $C(S')$, while Clarkson's Corollary~4.3~\cite{Clarkson87} deals with $T(S')$.
Let $R$ be an $r$-element random subset of $S$, and let $T^t_{\leq\frac{n}{r}}(R)$ be the set of configurations in $T(R)$ that conflicts with $t$ objects in $R$, but at most $\frac{n}{r}$ objects in $S$. In our terminology, \cite[Corollary~4.3]{Clarkson87} can be interpreted as follows:
\begin{equation}\label{eq-clarkson}
E[|T^t_{\leq \frac{n}{r}}(R)|]\leq O\big((\frac{e}{t})^t\big)\cdot |T(R)|.
\end{equation}

In the proof of Theorem~\ref{thm-cutting-size}, in order to analyze the expected size of our $\frac{1}{r}$-shallow cutting,
$C(R)$ is the set of pseudo-prisms in the vertical decomposition of $R$ (defined in Section~\ref{sub-sc-distance-VD}), and a surface conflicts with a pseudo-prism if the surface lies fully below the pseudo-prism. We need to bound $\sum_{t\geq 0} (1+t^4)\cdot E[|C^{t}_{\leq \frac{n}{r}}(R)|]$. Since $|C^{l}(R)|=O\big(r\cdot (l+1)^3\big)$ (\cite[Lemma~5.1]{KaplanMRSS17}), i.e., $|C^l(R)|=O(r)$ if $l=O(1)$, Theorem~\ref{thm-many-local-less-global} (with $c=1$ and $d=10$) implies a bound $\sum_{t\geq 0} \frac{1+t^4}{t!}\cdot O(r)=O(r)$. 
However, if we want to apply \cite[Corollary~4.3]{Clarkson87}, i.e., Inequality~(\ref{eq-clarkson}),
we could only use $E[|T^t_{\leq \frac{n}{r}}(R)|]$.
Moreover, since a pseudo-prism is defined by at most 10 surfaces (as shown in Section~\ref{sub-sc-distance-VD}),
we only have $T(R)=O(r^{10})$.
Therefore, 
a direct application of Inequality~(\ref{eq-clarkson}) would only obtain a bound $\sum_{t\geq 1}\frac{1+t^4}{t^t}\cdot O(r^{10})=O(r^{10})$.

In a chronological order, Clarkson~\cite{Clarkson87} first studied $T(S')$, then Clarkson and Shor~\cite{ClarksonS89} and Chazelle and Friedman~\cite{ChazelleF90} worked on $C^0(S')$ and considered an if-and-only-if condition,
\begin{itemize}
	\item[($\ast$)] $\triangle\in C^0(S')$ if and only if $D(\triangle)\subseteq S'$ and $K(\triangle)\cap S'=\emptyset$,
\end{itemize}
Agarwal et~al.~\cite{AgarwalMS98} and de~Berg et~al.~\cite{deBergDS1995} further relaxed the if-and-only-if condition by two weaker conditions
\begin{enumerate}[label=(\roman*)]
	\item for any $\triangle\in C^0(S')$, $D(\triangle)\subseteq S'$ and $K(\triangle)\cap S'=\emptyset$,
	\item and if $\triangle\in C^0(S')$ and $S''\subseteq S'$ with $D(\triangle)\subseteq S''$, then $\triangle\in C^0(S'')$,
\end{enumerate}
and we deal with $C^t(S')$ for which we generalize the above two conditions to include nonzero local conflicts with $t$ and $t'$ as follows,
\begin{enumerate}[label=(\Roman*)]
	\item For any $\triangle\in C^t(S')$, $D(\triangle)\subseteq S'$ and $|K(\triangle)\cap S'|=t$.
	\item If $\triangle\in C^t(S')$ and $S''\subseteq S'$ with $D(\triangle)\subseteq S''$ and $|K(\triangle)\cap S''|=t'$, 
	then $\triangle\subseteq C^{t'}(S'')$.
\end{enumerate}
At a high-level,
Clarkson and Shor~\cite{ClarksonS89}, Chazelle and Friedman~\cite{ChazelleF90}, Agarwal et~al.~\cite{AgarwalMS98}, and de~Berg et~al.~\cite{deBergDS1995} derived the probability that a configuration conflicts with no local object, but at least $t\frac{n}{r}$ objects, while we analyze the probability that a configuration conflicts with $t$ local objects, but at most $\frac{n}{r}$ objects. Clarkson~\cite{Clarkson87}'s work is more general and includes both cases, while it considers $T(S')$ instead of $C(S')$.

\bibliographystyle{plainurl}
\bibliography{KNN}

\newpage
\listoftodos 


\end{document}